\documentclass[12pt,a4paper]{article}

\usepackage{amsthm}
\usepackage{amsmath}
\usepackage{mathtools}
\usepackage{amssymb}
\usepackage{cite}
\usepackage{ulem}
\usepackage{stackrel}
\usepackage{enumerate}
\usepackage[table,xcdraw]{xcolor}
\usepackage[printonlyused]{acronym}
\usepackage{dsfont}
\usepackage{xfrac}

\newtheorem{lem}{Lemma}

\newtheorem{thm}{Theorem}

\theoremstyle{definition}
\newtheorem{remark}{Remark}
\def\eps{\epsilon}
\def\d{d}
\def\z{z}

\def\v{\bld{\phi}}
\def\vm{\phi_\k}
\def\w{\bld{\phi}}

\def\j{\ell}
\def\wG{\t(G,\F)}
\def\wS{\t(\S,\F)}
\def\wGi{t(G,\grH_\k)}
\def\wSi{t(\S,\grH_\k)}
\def\wSstar{\w(\S^*)}
\def\bld{\boldsymbol}
\def\G{\mathcal{G}}

\def\E{\mathbb{E}}
\def\F{\mathcal{F}}
\def\S{S}

\def\t{\bld{t}}
\def\p{p}
\def\grH{F}
\def\J{J}
\def\k{m}
\def\xmom{\E X}

\def\xmomi{\E X^\k}
\def\ymomi{\E Y^\k}

\def\pGF{p_G^F}
\def\pGFhat{\hat{p}_{G}^{\grH}}
\def\bH{b_{\grH,n}}
\def\cH{c_{\grH,n}}
\def\cHi{c_{\grH^\k,n}}

\def\KS{\mathsf{KS}}
\def\MP{\mathsf{\Psi}}
\def\hist{\mathsf{Hist}}
\def\ci{c_{m}}
\def\ctilde{5}
\def\twoctilde{10}

\def\Bveps{B_{\F}^n(\bld{\phi},\bld{\gamma})}

\def\Bvepsd{B_{\F^d}^n}

\def\Cw{C_{\F}^n(\bld{\phi},\bld{\gamma})}

\usepackage{scalerel,stackengine}
\stackMath
\newcommand\reallywidehat[1]{%
\savestack{\tmpbox}{\stretchto{%
  \scaleto{%
    \scalerel*[\widthof{\ensuremath{#1}}]{\kern-.6pt\bigwedge\kern-.6pt}%
    {\rule[-\textheight/2]{1ex}{\textheight}}
  }{\textheight}%
}{0.5ex}}%
\stackon[1pt]{#1}{\tmpbox}%
}
\DeclareMathOperator*{\argmax}{arg\,max}

\newif\ifarxiv
\arxivtrue

\begin{document}
	
\title{On the Number of Graphs with a Given Histogram}
	\author{Shahar~Stein~Ioushua and Ofer~Shayevitz\thanks{The authors are with the Department of EE--Systems, Tel Aviv University, Tel Aviv, Israel \{steinioushua@mail.tau.ac.il, ofersha@eng.tau.ac.il\}. This work was supported by the Israel Science Foundation (ISF) under grant 1495/18 and grant 1766/22. This paper was presented in part at 56th Annual Conference on Information Sciences and Systems (CISS) and the 2022 IEEE International Symposium on Information Theory (ISIT).}
	}
	\maketitle
	\begin{abstract}
    Let $G$ be a large (simple, unlabeled) dense graph on $n$ vertices. Suppose that we only know, or can estimate, the empirical distribution of the number of subgraphs $\grH$ that each vertex in $G$ participates in, for some fixed small graph $\grH$. How many other graphs would look essentially the same to us, i.e., would have a similar local structure? In this paper, we derive upper and lower bounds on the number of graphs whose empirical distribution lies close (in the Kolmogorov-Smirnov distance) to that of $G$. Our bounds are given as solutions to a maximum entropy problem on random graphs of a fixed size $k$ that does not depend on $n$, under $d$ global density constraints. The bounds are asymptotically close, with a gap that vanishes with $d$ at a rate that depends on the concentration function of the distribution at the center of the Kolmogorov-Smirnov ball. %
	\end{abstract}

\section {Introduction}\label{sec:intro_paper}
Let $G$ be a simple unlabeled dense graph on $n$ vertices. Suppose that we cannot access the entire graph, which could be very large, but we are nevertheless interested in gaining some information about its local structure. To that end, we could for instance randomly query a small number of vertices, and probe their local neighborhoods. For example, if we record the degrees of these random vertices, we can obtain a coarse estimate of the degree distribution of the graph, which we can think of as a \textit{histogram} of the true underlying (empirical) degree distribution. Similarly, we can obtain a histogram of the distribution of the number of triangles a vertex participates in. Or more generally, fixing any rooted graph $\grH$ on $r\ll n$ vertices, we can obtain a histogram of the \textit{$\grH$-degree distribution} of $G$, i.e., of the number of times a vertex in $G$ appears as the root of a subgraph $\grH$. We will refer to this as an \textit{$\grH$-histogram} of $G$. Given such an $\grH$-histogram, how much is revealed about the graph itself? More concretely, we are interested in characterizing the number of graphs whose $\grH$-histogram is similar to that of $G$.

\textbf{Contribution and techniques.} In this paper, we formalize the above question in a more abstract manner, by defining a histogram as a ball in the Kolmogorov-Smirnov (KS) metric around some smooth reference distribution. We then characterize the number of graphs whose true $\grH$-degree distribution lies inside the ball, in terms of a solution to a constrained maximum entropy problem over fixed-dimension random graphs. Our approach is based on the following ideas. First, we show that the local structure conditions can be essentially replaced by global ones with only a small penalty; loosely speaking, we show that the set of graphs with $\grH$-degree distribution inside a KS-ball roughly correspond to a set of graphs satisfying certain $d$ \textit{global density constraints} w.r.t. graphs $\{\grH^\k\}_{\k=1}^d$ derived from $\grH$, where by density we mean the relative occurrence of each $\grH^\k$ inside $G$. This reduction to vector of densities is obtained by counting arguments combined with an anti-concentration inequality. 

Given this reduction, it is enough to characterize the number of graphs with (roughly) a given density vector. To that end, and in a way somewhat reminiscent of the method-of-types~\cite{Csiszar_information_1982}, we define something we call a \textit{Szemer\'edi type} of a graph. A Szemer\'edi type is essentially a random graph $S$ over $k$ vertices with edges drawn independently with probabilities $s_{ij}$, together with some accuracy parameter $\eps$. Loosely speaking, a graph $G$ on $n\gg k$ vertices is said to have this Szemer\'edi type, if we can  partition (most of) its vertex set into $k$ equi-sized parts, such that the edges between parts $i$ and $j$ appear to have been drawn independently at random with probability $s_{ij}$, as long as we do not look too closely (determined by $\eps$). In general, a graph $G$ can have no Szemer\'edi type, or can have multiple Szemer\'edi types. The celebrated Szemer\'edi regularity lemma~\cite{szemeredi1975regular} implies that for $n$ large enough, every graph has at least one Szemer\'edi type, and hence one can roughly think of large enough graphs as being ``pseudorandom'' with a small number of parameters. We use this fact in order to convert our vector of densities problem into an optimization problem over (fixed dimension) Szemer\'edi types, and show that its solution is essentially given by the Szemer\'edi type with maximum entropy satisfying the corresponding expected density constraints. The gap between our upper and lower bounds depends on the number of density constraints $d$ we take, and vanishes at a rate that depends on the concentration function of the distribution at the center of our KS-ball.

\textbf{Motivation.} Graphs are widely used to describe relations between different elements in a system,  with notable examples including computer networks, social media, molecular structures and biological systems. Moreover, often these large networks have a typical underlying structure, such as edge-degree distribution (as observed e.g., for the Internet \cite{faloutsos1999power}) or triangle-degree distribution (e.g., the edge-triangle configuration model \cite{newman2009random,miller2009percolation}), and more generally, $\grH$-degree distribution \cite{karrer2010random,ritchie2014higher,ritchie2017generation}. While our work addresses a fundamental combinatorial question regarding the number of graphs with a certain distribution of local properties, it is primarily motivated by the need to provide a combinatorial framework for information theoretical problems over large graphs, in the spirit of the well-known typicality \cite{cover1999elements} and method-of-types~\cite{Csiszar_information_1982} used for sequences.

In the method-of-types, the basic idea is to characterize a finite sequence by the empirical distribution of its symbols, referred to as the \textit{type} of the sequence, and then quantify the number of sequences with the same type, which turns out to be exponential in the entropy of the type. This very simple idea goes a very long way in information theory both in data compression and communication, see \cite{Csiszar_information_1982}. In our setting, the simple sequence type is replaced by a graph $\grH$-histogram (that can be computed using one of the various existing subgraph listing or subgraph local counting algorithms, e.g., \cite{schank2005finding,kloks2000finding,chu2011triangle,chiba1985arboricity,becchetti2008efficient,becchetti2010efficient,suri2011counting}), or more generally and loosely speaking, by a Szemer\'edi type. In a way analogous to the classical result for sequences, the entropy of the Szemer\'edi type will yield the correct exponential behavior of the number of graphs that share the underlying histogram(s), with the small but important distinction that we allow a small deviation (in the KS metric). Hence our result in this sense is more reminiscent of sequence typicality \cite{cover1999elements} (we later discuss why exact counting is nontrivial in the graph setup, in contrast to the sequence setup). Just as in the sequence case, this immediately leads to a lossless compression result, describing the minimal number of bits required in order to represent a graph with a given histogram so that it can be exactly reconstructed. Furthermore, again as in the sequence case, this type of fundamental characterization can be leveraged to address more difficult information-theoretic questions.

One such interesting example is structure-preserving lossy compression of graphs. It has been long observed \cite{watts1998collective} that the global properties of a graph, e.g., edge probability and subgraph counts, are not enough in order to describe realistic networks. Local properties, such as the clustering coefficient and the small world property are required \cite{watts1998collective,albert2002statistical}. 
These properties measure the dependence between the edge probability of two vertices and their joint neighborhood. For example, large clustering coefficient suggests that if two vertices have many joint neighbors, they are likely to have an edge between them.  Therefore, preserving the local neighborhood of each vertex in the compression process can be desirable. If one uses a global distortion measure, as, say, the Hamming distance, then even a very small overall distortion (i.e., edges removed or added) can concentrate on a small subset of vertices or even erase a specific neighborhood entirely. This can be prevented by using a  measure that controls the local distortion. In this type of a setting, we are given a large graph, and would like to represent it using the minimal number of bits, while preserving some local structure up to some prescribed distortion, for each and every vertex (e.g., the set of actual triangles in which each vertex participates in the original graph and its reconstruction, should have say a $95\%$ overlap). This type of problem was previously considered in~\cite{bustin2021lossy} for the edge case, by using known results on edge-degree counts \cite{barvinok_number_2010,barvinok_matrices_2012}, see Section~\ref{sec:related_work} for further details. In order to generalize this method for measures that control the $\grH$-degree distortion for general $\grH$, results on the number of graphs with (approximately) the same histogram, which we give here, appear to be required.

\textbf{Organization.} Section~\ref{sec:related_work} reviews relevant previous results and related work. Section~\ref{sec:perliminaries} includes some necessary preliminaries and definitions. In Section~\ref{sec:problem_setting} we formally present the graph histogram characterization problem and state our main results. In Section~\ref{sec:szemerdi_types} we introduce the Szemer\'edi types, which play a central role in our proofs. Then, in Section~\ref{sec:proof_of_main_result} we prove our main result in two steps: First, in Subsection~\ref{sec:reduction_to_densities}, we reduce the histogram characterization problem to a vector subgraph densities problem via anticoncentration arguments. Then, in Subsection~\ref{sec:sol_of_densities_problem}, we resolve the densities problem via Szemer\'edi types and the regularity lemma, thereby yielding our main result. In Section~\ref{sec:single_density_solution}, we slightly digress to examine the subgraph densities problem in the special case of a scalar density; unlike the general vector setting, this case admits a simple finite $n$ expression. Finally, in Section~\ref{sec:discussion}, we conclude by discussing some gaps as well as possible extensions of our work. Additional background and some technical proofs are relegated to the  Appendix.

\section{Related work}\label{sec:related_work}

The case of edge-degree distributions, where $\grH$ is a single edge on $r=2$ vertices, was solved by Barvinok~\cite{barvinok_number_2010,barvinok_matrices_2012}, who used a generating functions approach to show that the number of graphs with a given degree distribution is obtained as a solution to an $n$-parameter maximum entropy problem. Compared to our approach, his result is stronger as it yields the number of graphs (up to sub-exponential terms) that have the \textit{exact} given edge-degree distribution, rather than close in the KS metric, albeit his solution involves $n$ dimensional optimization rather than $k\ll n$ dimensional in our case (but this can of course be fixed). Recently, the authors gave a simpler information-theoretic proof for Barvinok's result by leveraging degree-distribution invariant operations~\cite{ioushua2019counting}.  Later in \cite{bustin_lossy_2017,bustin2021lossy}, these results were used to derive bounds on the rate distortion function for universal lossy compression of directed random graphs, under distortion measure that limits the number of edges one vertex can lose or gain in the compression process. 
The  degree sequence of graphs was also studied in \cite{chatterjee2011random}, were it was shown that sequences of graphs with a given degree sequence converge to a unique, identifiable limiting object. We elaborate on this limiting object, called a \textit{graphon}, in Subsection ~\ref{sec:intro_CV_bounds}.

The general case of $\grH$-degree distributions or $\grH$-histograms has not been addressed before, and it appears that both the approaches used by Barvinok and the authors are not easy to extend. The difficulty of going beyond the edge case is underscored by the \textit{feasibility problem}, i.e., checking whether the set of graphs with a given $\grH$-degree distribution is nonempty. The feasibility problem is resolved by the Gale-Ryser Theorem~\cite{ryser_term_1958,gale_theorem_1957} in the edge case, but is not well understood for general $\grH$. This is also another reason to work with $\grH$-histograms rather than exact $\grH$-degree distributions. 

The problem of characterizing the number of graphs with an approximately given scalar graph parameter (as defined later), in the limit of $n\rightarrow\infty$, was addressed by Chatterjee and Varadhan~\cite{chatterjee2011large} using an object called a graphon\cite{lovasz2012large}, that can be loosely thought of as a compact limit of Szemer\'edi types for $k\rightarrow \infty$. They derived a large deviation result that yields upper and lower bounds on the exponent of the size; in particular relevance to our setting, they showed that their bounds coincide when the graph parameter is a \textit{scalar} subgraph density, yielding a size exponent given by the maximum entropy graphon satisfying the single density constraint.
In \cite{chatterjee2013estimating}, these results were used to approximate the normalizing constant of exponential random graph models with subgraph counts as their sufficient statistics.
In~\cite{lubetzky2015replica}, Lubetzky and Zhao studied related problems via the regularity lemma lens, and in particular gave a more general sufficient condition on the scalar parameter for the bounds to coincide. These works however do not readily generalize to the vector densities case, as outlined in the next subsection (which can be skipped without affecting the rest of the paper).


\subsection{Graphons and the Chatterjee-Varadhan (CV) Large Deviation Bounds}\label{sec:intro_CV_bounds}

In a sequence of papers (see~\cite{lovasz2012large} and references therein), Lova\'sz et al. developed a beautiful theory of graph limits, and introduced the notion of graphons. 
A graphon is a measurable function $f:[0,1]\rightarrow[0,1]$ that satisfies $f(x,y)=f(y,x)$. The space of graphons (up to measure preserving bijections) can be endowed with a metric called the {\em cut metric} (see Appendix~\ref{appendix:CV_bounds}). Any graph $G$ on $n$ vertices can be represented as a ``checkerboard'' graphon $f^G$, where
\begin{align}\label{eq:fG}
    f^G(x,y) = \begin{cases}1,&\;\text{if $\lceil nx\rceil\sim\lceil ny\rceil$ is an edge in $G$}\\ 0,&\;\text{otherwise}\end{cases}.
\end{align}
A \textit{graph parameter} is a (scalar or vector) function $\tau(G)$ that  can be extended to a continuous (w.r.t to the cut metric) function $\tau(f)$ over graphons, such that $\tau(f^G) = \tau(G)$. Graphons become useful in the context of graphs due to the following property: Let $G_n$ be a sequence of simple graphs over $n$ vertices, and let $t(\grH,G_n)$ denote the relative occurrence of $\grH$ in $G_n$ (a formal definition of $t(\grH,G)$ will later appear). 
Then, if for any fixed $\grH$ the sequence $t(\grH,G_n)$ converges as $n\to\infty$, the sequence $G_n$ converges to a graphon $f$, in the sense that $f^{G_n}$ converges to $f$ in the cut metric. 

Chatterjee and Varadhan used Lov\'asz's graphon framework to derive a large deviation result~\cite[Theorem 2.3]{chatterjee2011large} that gives bounds on the number of graphs $G$ whose parameter $\tau(G)$ lies inside a given set, in the limit of large $n$. Their bounds  (stated explicitly in Appendix~\ref{appendix:CV_bounds}) are given as a maximum entropy problem over a set of graphons that satisfy the associated constraints on their corresponding graph parameter. The lower bound corresponds to the maximum over the set of graphons whose parameter value lies within an open set, and the upper bound correspond to the maximum over the set of graphons whose parameter value lies within the closure of that set. Formal definitions and the explicit bounds can be found in Appendix~\ref{appendix:CV_bounds}.

The aforesaid CV large deviation result is quite general and can be applied to various graph parameters, scalar or vector, as also remarked by the authors in \cite{chatterjee2011large}. However, it is important to note that the resulting bounds are not generally tight; namely, the upper and lower CV bounds coincide in some cases but not in others. Specifically, the CV bounds are not tight in the case where the maximum entropy solution is not continuous w.r.t the parameter level. If an infinitesimally small change in the value of the parameter that defines the feasible set of graphons can cause a large change in the corresponding maximum entropy, a jump discontinuity will occur at the boundary of (some) closed set that defines the CV upper bound, thus strictly separating it from the CV lower bound. An explicit (toy) example of a scalar graph parameter for which the bounds are not tight can be found in Appendix~\ref{appendix:CV_bounds}. In the special case where the parameter is the $\grH$-density of the graph, $\tau(G)=t(\grH,G)$, the bounds are indeed tight and thus yielding (asymptotically) the exact exponential count of the number of graphs with $\grH$-density inside a given interval. 

  The above tightness issue has been addressed in \cite{lubetzky2015replica}, in the scalar parameter case. A graph parameter $\tau(f)$ is called {\em nice} when all its local extrema are also global ones. The authors in \cite{lubetzky2015replica} showed that the CV bounds (applied e.g. to an open interval and its closure) are tight for all (scalar) nice graph parameters. %
  Specifically, define the {\em rate function} $\psi(t)$ as the solution to a maximum entropy problem 
  \begin{align}\label{eq:rate_func}
   \psi(t) \triangleq \sup \{ H(f), \text{over all graphons $f$ such that } \tau(f)\geq t\},  
  \end{align}
  where 
  \begin{align}
      H (f) = \int_{[0,1]^2} h(f(x,y))dxdy,
  \end{align}
  and $h(\cdot)$ is the binary entropy function. Then $\tau(f)$ being nice implies that the rate function is continuous in $t$, which is a necessary and sufficient condition for the bounds to  be tight.
 The $\grH$-density of the graph was shown to be nice in \cite{lubetzky2015replica}, and in fact the continuity of $\psi(t)$ in this case has already been implicitly used by Chatterjee and Vardahn in the proof of \cite[Theorem 4.1]{chatterjee2011large}.

  To the best of our knowledge, there has been no prior work giving conditions for the tightness of the CV bounds in the case that $\tau(f)$ is a vector parameter. The difficulty in this case is to generalize the notion of maxima and minima to vectors in a way that yields a property equivalent to niceness in the scalar case.  
 Part of our contribution here is to suggest a generalization of niceness in multiple dimensions; we do so by defining {\em local and global boundary points} (formally introduced in Subsection~\ref{sec:sol_of_densities_problem}) and requiring that every local boundary point is also a global one.  However, this property is difficult to generally establish, as discussed also in Subsection~\ref{sec:sol_of_densities_problem}. 
 We circumvent this issue by defining a local measure of niceness, the {\em effective radius}, in a finite $n$ setup, which will imply that there are no boundary points in a specific neighborhood (in the Szemer\'edi types space). 

 In fact, it should be further stressed that niceness of the parameter (in both the scalar and vector cases) is not sufficient in our setup, since we are interested in a finite $n$ characterization. The CV bounds work on asymptotical objects (graphons), and do not immediately yield finite $n$ bounds.  In order to obtain such bounds, one needs to use the Szemer\'edi regularity lemma, and then a suitable counting lemma.  The resulting bounds are then given as maximum entropy problems over sets of Szemer\'edi types.  Then, loosely speaking, in order for the bounds to be tight we need the maximum entropy Szemer\'edi type of the lower and upper bounds not to be too far apart.  This requires, in addition to continuity of the maximum entropy solution with respect to the parameter level, also some quantitative result bounding the gradient of the mapping $\tau(G)$, as further elaborated in Subsection~\ref{sec:sol_of_densities_problem}. 

We remark that as observed in~\cite{radin2013phase}, the CV large deviation bounds coincide for \textit{any} graph parameter (scalar or vector) in the case where one is only interested in the asymptotic size exponent for an \textit{exact}  level of the graph parameter. This asymptotic exponent was studied in~\cite{radin2013phase} for the edge-triangle case, and later generalized in~\cite{kenyon2017multipodal} to a vector of $k$-star densities. While this type of analysis is insightful and interesting, it should be noted that the resulting asymptotic size exponents have no operational meaning due to order of limits (first in $n$, and then in the ball size around the parameter level). For instance, it could be that there exist no graphs (for any $n$) at certain parameter value (e.g., irrational density) but the asymptotic size exponent is still nontrivial. To obtain operational results, the graph parameter must therefore satisfy the above mentioned continuity properties. Furthermore, these properties need to hold in a stronger sense in order to obtain a finite $n$ characterization, which is  necessary for example in applications such as lossy and lossless compression of graphs discussed in Section~\ref{sec:intro_paper}, and is the approach we take here.


\section{Preliminaries and Notation} \label{sec:perliminaries}
The \textit{Kolmogorov-Smirnov (KS) distance} between two probability distributions $p$ and $q$ over $[0,1]$, is the $L^\infty$ distance between c.d.f.s, i.e., 
\begin{align}
    \KS(p,q) = \sup_{x\in[0,1]}\left| \int_0^x dp - \int_0^x dq\right| .
\end{align}
The concentration function of $p$ returns the maximal probability that $p$ gives to any interval of a given size, i.e., 
\begin{align}\label{eq:concent}
S_p(a) = \sup_{x} \int_{x}^{x+a}dp.
\end{align}
We use boldface letters to indicate vectors and write $\boldsymbol{u}\leq \boldsymbol{v}$ for $\boldsymbol{u},\boldsymbol{v}\in\mathbb{R}^d$ to mean an entry--wise inequality. 
In this work, all graphs are simple, undirected, and unlabeled, unless otherwise noted. For a graph $G$ we denote by $V(G)$ the set of all vertices in $G$. We write $i\sim j$ for two vertices $i,j\in G$ to mean that there is an edge between $i$ and $j$ in $G$. The collection of all  graphs on $n$ vertices is denoted by $\G_n$. For a random graph $\S$ the notation $S\sim (s_{ij})$ means that $\S$ has independent edges and the edge $(i,j)$ exists with probability $s_{ij}$. If $\S$ is a matrix then we denote by $[\S]_{ij}$ the element in the $i$th row and $j$th column. In all the following we use the natural base for the log function.


\section {Problem Setting and Main Result}\label{sec:problem_setting}
Let $G$ be a graph on $n$ vertices. Let $\grH$ be a graph  on $r \ll n$ vertices, where one of the vertices is designated as the {\em root vertex}. 
We say that a vertex $v$ of $G$ has {\em $\grH$-degree} $i$, if $v$ appears as the root vertex in exactly $i$ \textit{copies} of $\grH$ in $G$.
Loosely speaking, a copy is any appearance of $F$ inside $G$, counted once. Precisely, a copy is an element in the quotient set $\Pi/ R$, where $\Pi$ is the set of all injections $\pi:V(\grH)\to V(G)$ such that $i\sim j$ in $\grH$ implies $\pi(i)\sim \pi(j)$ in $G$, and $R$ is an equivalence relation on $\Pi$ where $\pi, \pi'\in\Pi$ are equivalent if $\pi=\sigma(\pi')$ for some $\sigma$ from the automorphism group of $\grH$.

Define the {\em$\grH$-degree distribution} $\pGFhat$ of the graph $G$ to be the empirical distribution of the vertices $\grH$-degrees, i.e.,  
\begin{align}
    \pGFhat(m) = \frac{1}{n}\sum_{v\in V(G)} \mathds{1} (v\textrm{ has $\grH$-degree } m). 
\end{align}
Furthermore, define \textit{normalized} $\grH$-degree distribution $\pGF$, obtained from $\pGFhat$ by normalizing the latter to have support in $[0,1]$. Namely,  letting $\bH$ denote the $\grH$-degree of a vertex in the complete graph on $n$ vertices (which is the maximal possible), then $X = \frac{\hat{X}}{\bH}$ with $\hat{X}\sim \pGFhat$ implies $X\sim \pGF$. For example, if $\grH$ is a clique of size $r$ we have $\bH={\binom{n-1}{ r-1}}$. 

For any distribution $\p$ with support in $[0,1]$, we define the {\em $\grH$-histogram} $\hist(\p,\grH,\delta,n)$ as the set of all graphs $G$ on $n$ vertices, whose normalized $\grH$-degree distribution $\pGF$ is $\delta$-close to $\p$ in KS distance, i.e., 
    \begin{equation}
        \hist(\p,\grH,\delta,n) = \{G\in\G_n : \KS(\pGF,\p)\leq\delta\}.
    \end{equation}
From here on, we suppress $n$ and write $\hist(\p,\grH,\delta)$ to denote $\hist(\p,\grH,\delta,n)$. In what follows, we will be interested in characterizing the $\grH$-histogram size. Specifically, we will show that the logarithm of the $\grH$-histogram size is approximately equal to the entropy of a random graph with independent edges on $k\ll n$ vertices, where the maximization is subject to suitable subgraph density constraints. 

To that end, define the $\grH$-density of a graph $G$ to be 
\begin{align}\label{eq:density_deterministic}
    t(G,\grH) \triangleq \frac{\text{number of  copies of $\grH$ in $G$}}{\cH},
\end{align}
where $\cH$ is the number of  copies of $\grH$ in the complete graph on $n$ vertices. The $\grH$-density $t(\S,\grH)$ of a random graph $\S\sim (s_{ij})$ can be defined similarly, to be (roughly) the expected fraction of appearances of $\grH$ in $S$; the exact expression for the mean density, which is a polynomial in $(s_{ij})$, is discussed later in Section~\ref{sec:szemerdi_types}. 
Given an indexed family of $d$ small graphs $\F=\{\grH_\k\}_{\k=1}^d$  we write $\t(\S, \F)$ to denote the corresponding vector of $\grH_\k$-densities. Let $\bld{\gamma}, \v\in [0,1]^d$, and define 
\begin{align}\label{opt:MP}
    \MP(\v,\bld{\gamma},\F,k) = &\max_{\S:  |V(S)| = k,\, |\t(\S,\F)-\v|\leq \bld{\gamma}}  \tfrac{1}{k^2}\cdot H(S)
\end{align}
where $H(S)=\sum h(s_{ij})$. Namely, $\MP(\v,\bld{\gamma},\F,k)$ is the maximal possible (per-edge) entropy of a random graph $\S$ on $k$ vertices, with $\grH_\k$-densities that are $\gamma_\k$-close to $\vm$. 

\newcommand{\Hi}[1]{\grH^{#1}}
Our main result shows that for a specific choice of $\v,\bld{\gamma}$, and $\F$, the function $\MP(\v,\bld{\gamma},\F)$ gives upper and lower bounds on the $\grH$-histogram size. To state our result, we need a few more definitions. First, let $\Hi{m}$ be a graph on $1+\k(r-1)$ vertices obtained by taking $m$ disjoint copies of $\grH$, and merging their root vertices into a single joint root vertex. Now, fixing some $d$ and the graph $\grH$, let 
\begin{align}\label{eq:F_set}
    \F^d \triangleq \left\{\Hi{1},\cdots,\Hi{d}\right\}.
\end{align}
Furthermore, set some reference distribution $\p$ and let $X\sim \p$, and $\v(\p)\in \mathbb{R}^d$ have entries 
\begin{align}\label{eq:density_approximation_out_of_p} 
    \vm(\p) \triangleq  \ci \cdot \xmom^\k,\quad  \k\in[d]
\end{align}
where 
\begin{align}\label{eq:asymptotic_normalization_ratio}
    \ci \triangleq {\displaystyle \lim_{n\rightarrow \infty }}\frac{n\bH^\k}{a_\k \cHi \k!}  \leq (\k(r-1)+1)^{\k(r-1)+1}, 
\end{align}
and $a_\k=|\mathrm{Aut}(\grH^\k)|$ is the size of the automorphism group of $\grH^\k$. The limit \eqref{eq:asymptotic_normalization_ratio} exists since $n\cdot \bH^\k$ and $\cHi$ are polynomials of the same degree in $n$. We later show that $\vm(\pGF)$ approximates the $\grH^\k$-density of $G$ with $O(n^{-1})$ error.
We will also show that the $\KS$-distance between two distributions can be bounded in terms of the closeness of their first $d$ moments (and therefore in terms of $\vm(\cdot)$, $\k=1,\cdots,d$), thus  translating the  $\grH$-degree distribution constraint to $d$ global density constraints. 

We are now ready to state our main result. For better readability, we omit some technical details at this point, and provide them later in Subsection~\ref{sec:sol_of_densities_problem} (see Theorem~\ref{thm:TypicalSetSizeGeneral}). 
\begin{thm}\label{thm:main_theorem}
Let $\p$ be a distribution on $[0,1]$ that is absolutely continuous w.r.t. the Lebesgue measure, with a density bounded away from zero and infinity. Then, under some regularity conditions, for $k=k(\delta)$ and any sufficiently large $n$, 
\begin{align*}
   \MP(\v(p),\boldsymbol{\beta},\F^d,k)& \lessapprox\hspace{-0.03cm}  \frac{1}{n^2}\log\left|\hist(\p,\grH,\delta)\right|
   \lessapprox \MP(\v(p),\boldsymbol{\gamma},\F^d,k)
\end{align*}
where $\|\boldsymbol{\beta}\|_\infty = \Omega_d(\delta e^{-1/\delta})$ and $\|\boldsymbol{\gamma}\|_\infty = O_d(\delta)$, and $\lessapprox$ means up to some vanishing (with $n$) error terms. 
\end{thm}


\section{Szemer\'edi Types}\label{sec:szemerdi_types}
We begin by defining two important entities: a {\em $(k,\eps)$-uniform partition} of a graph and a {\em Szemer\'edi type} of a graph. We will later use these notions in Subsection~\ref{sec:sol_of_densities_problem} when we solve the maximum entropy problem with density constrains. Let $G$  be a graph and $A,B\subset V(G)$ be a pair of disjoint subsets. The density of the pair $(A, B)$ is the fraction $d(A,B)=\frac{e(A,B)}{|A||B|}$ where $e(A,B)$ is the number of edges with one endpoint in $A$ and second in $B$ and $|A|,|B|$ denote the cardinalities of $A$ and $B$, respectively. The pair $(A,B)$ is called $\eps$-uniform if for every $A'\subseteq A$, $B'\subseteq B$, $|A'|\geq\eps |A|$, $|B'|\geq\eps |B|$, it holds that $|d(A',B') - d(A,B)|<\eps$. A partition $V(G) = C_0  \cup C_1 \cup \cdots \cup C_k$ is called $(k,\eps)$-uniform if 
    \begin{enumerate}
        \item $|C_0| < \eps |V(G)|$. 
        \item $|C_1|=|C_2|=\cdots =|C_k|$.
        \item all but $\eps\binom{k} {2}$ of the pairs $(C_i,C_j)$ are $\eps$-uniform. 
    \end{enumerate}

    A \textit{Szemer\'edi type} is a triplet $(k,\eps,\S)$, where $k\in\mathbb{N}$, $\eps\in[0,1]$, and $S$ is a $k\times k$ symmetric matrix with entries in $[0,1]$. We say that a graph $G$ has Szemer\'edi type $(k,\eps,\S)$ if there exists some $(k,\eps)$-uniform partition of $G$ and $[\S]_{ij} = d(C_i,C_j)$ for $i\neq j\in[k]$. From hereon we denote $[\S]_{ij}\triangleq s_{ij}$. We will loosely refer to the matrix $\S$ as a $(k,\eps)$-Szemer\'edi type, or simply a Szemer\'edi type, when  $k,\eps$ are clear from context. The {\em Szemer\'edi type class} $\Lambda(k,\eps,\S,n)$ is defined as the collection of all graphs on $n$ vertices for which $\S$ is a $(k,\eps)$-Szemer\'edi type. We will sometimes write $\Lambda(\S,n)$ when $k,\eps$ are clear from context. 
     
    Note that a graph can have multiple $(k,\eps)$-Szemer\'edi types, or none at all. The following well-known lemma shows that when $n$ is sufficiently large, at least one such type exists. 
    
    \begin{lem}[Szemer\'edi regularity lemma \cite{szemeredi1975regular}] For any $\eps > 0$, there exist positive
integers $n_0=n_0(\eps)$, $m_0=m_0(\eps)$ and $k=k(\eps)$, $m_0\leq k\leq n_0$,  such that every graph with at least $n_0$ vertices has  at least one $(k,\eps)$-Szemer\'edi type.
\end{lem}\label{lem:SzemerLem}

We remark that the $k$ in the above lemma is usually very large. It is easily lower bounded by $k \geq {1}/{\eps}$, but is in fact much larger; the best known upper bound is a tower of exponentials of height proportional to $\eps^{-5}$ \cite{szemeredi1975regular}, and this cannot be significantly improved~\cite{gowers1997lower}. However, the important point is that  $k$ is a function of $\eps$ only, and does not depend on $n$.

    A Szemer\'edi partition is essentially a coarse partition of the vertex set that ``looks random''.  A Szemer\'edi type $(k,\eps,\S)$ therefore naturally corresponds to a random graph $\S\sim (s_{ij})$ over $k$ vertices, and we can think of $t(\S,\grH)$ as essentially being the probability of seeing a  copy of $\grH$ when looking at $r$ uniformly random vertices of the random graph $\S$. More accurately, let $J = (J_1,\ldots J_r)$ be uniformly distributed over $[k]^r$, and $\S|_{J}$ be a graph with vertex set $[r]$, where $i\sim j $ in $\S|_{J}$ if and only if $J_i\neq J_j$ and $J_i \sim J_j$ in $\S$. Then, we define  
    \begin{align}
        t(\S,\grH)\triangleq \Pr(\grH\text{ is a subgraph of }S|_{J}).
    \end{align}

Note that the above extends the deterministic subgraph density~\eqref{eq:density_deterministic}, with the small (but significant) distinction of allowing multiple drawings of the same vertex. 
The exact expression of $t(\S,\grH)$ is a polynomial of degree at most $\binom{r}{ 2}$.
In the simple case that $\grH$ is clique it is given by
\begin{align}
    t(\S,\grH)\triangleq \frac{1}{\binom{k}{r}} \sum_{A\subseteq [k], |A|=r}\;\prod_{i<j\in A}s_{ij}.
\end{align} 
To generalize the above to non-cliques, we need some adjustments.  
Now, vertices in $\grH$ that are not connected with an edge between them can be chosen in the same set $C_i$. Hence, the mean $\grH$-type depends on all the different possible {\em proper vertex colorings} of $\grH$. A proper vertex coloring of a graph is a labeling of the graph’s vertices with colors such that no two vertices sharing the same edge have the same color.
We define by $\mathcal{X}$ the set of all proper coloring of $\grH$. Then, $|X|$ denotes the number of different colors in the coloring $X\in\mathcal{X}$, and we have that $ \chi \leq |X| \leq r$, where $\chi$ is the {\em chromatic number} of $\grH$. The occurrence vector of a coloring $X$, is a vector of length $|X|$ whose $\ell$th entry records the number of vertices that has the $\ell$th color in the coloring. 
We call two coloring different if their corresponding occurrence vectors are not equal in some coordinate. 
Define the mean $\grH$-density of a $(k,\eps)$-Szemer\'edi type $S$ to be
\begin{equation*}
    t(\S,\grH) \triangleq \frac{1}{a_1\binom{k}{r}}\sum_{p=\chi}^r  \sum_{_{X\in\mathcal{X}:\; |X|=p}} \sum_{\small{\begin{subarray}{l}A\subseteq [k]\\
    A=\{i_j\}_{j=1}^{|X|}\end{subarray}}}\prod_{\small{\begin{subarray}{l} \;\;\;\ell<j\in [|X|]\\ \;\;x(\ell)\neq x(j)\end{subarray}}}\hspace{-0.4cm}s_{i_{\ell}i_{j}},
\end{equation*} 
where $a_1=|\mathrm{Aut}(\grH)|$  is the size of the automorphism group of $\grH$.
In the above, the outer sum is over all the possible sizes $p$, $\chi\leq p\leq r$ of coloring of $\grH$, and the inner sum runs over all the proper coloring $X$ with size $|X|=p$.

In Lemma~\ref{lem:counting_lemma}, we will show that all graphs $G\in\Lambda(\S,n)$ have $\grH$-density that is approximately equal to $t(\S,\grH)$, and moreover, have bounded deviation from this average quantity.  
  
\section{Proof of Main Result}\label{sec:proof_of_main_result}
Let $\F$ be a family of $d$ graphs. Let $\v, \bld{\gamma}\in [0,1]^d$. Then, the {\em $\F$-densities set} $\Bveps$, is the set of all graphs on $n$ vertices whose $\F$-densities vector is $\bld{\gamma}$-close to $v$, i.e., 
\begin{align} \label{eq:typical_set_Bv}
    \Bveps \triangleq \{G\in\G_n:|\t(G,\F)-\v|\leq \bld{\gamma}\}
\end{align}

Recall the specific definition of the family $\F^d$ in ~\eqref{eq:F_set}, induced by a single graph $F$. The first step in proving our main result is showing that for $\v = \v(p)$ defined in~\eqref{eq:density_approximation_out_of_p}, there exist a choice of $\bld{\gamma}=\bld{\gamma}(\p,\delta,d)$ and $\bld{\beta}=\bld{\beta}(\p,\delta,d)$, such that $\Bvepsd(\v,\bld{\beta}) \subseteq \hist(\p,\grH,\delta)\subseteq \Bvepsd(\v,\bld{\gamma})$. This will be done in Subsection~\ref{sec:reduction_to_densities}. 
 Then, in Subsection~\ref{sec:sol_of_densities_problem} we give sufficient conditions on $\v$ under which the size of the set $\Bvepsd(\v,\bld{\gamma})$ is (asymptotically) equal to the maximum entropy solution $\MP(\v,\bld{\gamma},\F^d,k)$ given in~\eqref{opt:MP}.


 \subsection{Reduction to Vector of Densities} \label{sec:reduction_to_densities}

We now show that the $\grH$-histogram problem can be reduced to a densities-type enumeration problem. This result is embodied in Lemmas~\ref{lem:upper_bound} and~\ref{lem:LowerBound}. To that end, we need two supporting lemmas. The first lemma shows that the $\F^d$-densities $\t(G,\F^d)$ can be approximated from the $\grH$-degree distribution $\pGF$.

\begin{lem}\label{lem:density_equal_moment}
Let $G$ be a graph with $\grH$-degree distribution $\pGF$ and $\F^d$-densities vector $\t(G,\F^d)$. Let $\v(G)$ be computed from $\pGF$ as in \eqref{eq:density_approximation_out_of_p}.
Then, 
\begin{align}
    \|\t(G,\F^d)-\v(\pGF)\|_\infty =O(1/n).
\end{align}
\end{lem}
\begin{proof}
We prove for $m>1$. The proof for $\k=1$ is the same except for the normalization constants. First, we show that one can approximate $t(G,\grH^\k)$ using
\begin{align}
    \tilde{t}(G,\grH^\k) = \frac{n}{\cHi}\cdot \E{\binom{\hat{X}}{m}}.
\end{align}
where $\hat{X}\sim \hat{p}(G,\grH)$ is the (un-normalized) $\grH$-degree r.v. 
Then, we show that $|\vm-\tilde{t}(G,\grH^\k)|=O(n^{-1})$. 
A \textit{bad copy of $\grH^\k$} is any graph $Q$ that consist of $\k$ copies of $\grH$ that share one root vertex, and has at least one more vertex that is common between the different copies. Hence $|V(Q)| <1+\k(r-1)$. Let $w_{\k,\text{bad}}$ denote the number of bad copies of $\grH^\k$ in $G$. 
Then 
\begin{align}
   \tilde{t}(G,\grH^\k) =t(G,\grH^\k) + \frac{1}{\cHi}w_{\k,\text{bad}}.
\end{align}
Next, note that since any such bad copy in $G$ involves choosing at most $\k(r-1)$ vertices out of $n$ and accounting for their different permutations, we get that  $w_{\k,\text{bad}}\leq \sum_{p=1}^{\k(r-1)}2^p \binom{n}{p}=O(n^{\k(r-1)})$. Then, since $\cHi=\Theta(n^{\k(r-1)+1})$ (see for example \cite{pippenger1975inducibility}) we get $$  \tilde{t}(G,\grH^\k)-t(G,\grH^\k)=\frac{1}{\cHi}w_{\k,\text{bad}} =O(n^{-1}).$$
Recall that
\begin{align}
    \binom{\hat{X}}{\k}=\sum_{\j=1}^\k \z_{\k,\j}\frac{\hat{X}^\j}{\k!}=\sum_{\j=1}^\k \z_{\k,\j}\frac{\bH^\j\cdot X^\j}{\k!}, 
\end{align}
where $\z_{\k,\j}$ are the Stirling numbers of the first kind  whose absolute value is monotonically decreasing in $\j$, $\z_{\k,\k}=1$ and $|\z_{\k,1}|=(\k-1)!$. Hence \begin{align}
   \tilde{t}(G,\grH^\k) &= \frac{n\bH^\k}{a_\k \cHi \k!}\left(\xmomi+\sum_{\j=1}^{\k-1} \frac{\z[\k,\j]}{\bH^{\k-\j}}\cdot \xmom^\j \right)\\
   &= \frac{n\bH^\k}{a_\k \cHi \k!}\left(\xmomi+O(n^{-(r-1)}) \right)\label{line:limit_of_const0}\\
   &= (\ci+O(n^{-1}))\left(\xmomi+O(n^{-(r-1)}) \right)\label{line:limit_of_const}
   \\ &= \vm + O(n^{-1}),
\end{align}
where in \eqref{line:limit_of_const0} and \eqref{line:limit_of_const} we used the fact that $\xmom^\j\leq 1$, $\bH^\k=\Theta(n^{\k(r-1)})$ and $\cHi=\Theta(n^{\k(r-1)+1})$. \end{proof}

Next, we show that if two distributions are close in KS distance, then their corresponding moment vectors  are also close in the $L_\infty$ distance, and vice versa.

\begin{lem}\label{lem:KS_moments_equivalence}
Let $p$ and $q$ be two distributions over $[0,1]$, and set $X\sim p$, $Y\sim q$.
If $\KS(p,q) \leq \delta$, then for every $\k\in \mathbb{N}$
\begin{align}
    \left|\ymomi-\xmomi\right| \leq  \delta. 
\end{align}
Conversely, if $\left|\ymomi-\xmomi\right| \leq \gamma$ for any $\k\in[d]$, then 
\begin{align}\label{eq:KS_upper_bound_function_of_T}
\KS(p,q) \leq c \left(S_p(1/T) + e^T\cdot \left(\gamma  +  \frac{T^{d+1}}{d!d}\right) \right).
\end{align}
where $c\leq 51$, $S_p(u)$ is the concentration function of $p$, and $T > 1$ a parameter that can be optimized. 
\end{lem}

\begin{proof}
Let $P(x)$ and $Q(x)$ be the c.d.fs that correspond to the distributions $p$ and $q$, respectively.  
Note that 
\begin{align*}
    \xmomi=\int_0^1p(x)x^mdx = P(x)x^m\Big|_0^1-m\int_0^1x^{m-1}P(x)dx = 1-m\int_0^1 x^{m-1}P(x)dx, 
\end{align*}
where we used integration by parts. Similarly,
\begin{align*}
    \ymomi=1-m\int_0^1 x^{m-1}Q(x)dx, 
\end{align*}
then
\begin{align*}
    |\xmomi&-\ymomi|=\left|\int_0^1 m x^{m-1}(Q(x)-P(x))dx\right|\\
   & \leq \int_0^1 m x^{m-1}|Q(x)-P(x)|dx\\
   & \leq \KS(p,q)\cdot \int_0^1 mx^{m-1}dx =\delta
\end{align*}
Conversely, suppose the $m$th moments are $\gamma$-close, for any $m\in[d]$. Let $\psi_q(t)$ and $\psi_p(t)$ be the characteristic functions of $p$ and $q$, respectively.
From Fainleib's generalization of Esseen's inequality \cite{fainleib1968generalization,esseen1945fourier} we get that for any $T>0$, 
\begin{align}\label{eq:esseen_fainleib}
    \KS(p,q) \leq  c\left( S_p(1/T)+\int_0^T\frac{|\psi_q(t)-\psi_p(t)|}{t}dt\right).
\end{align}
Now, we can express the characteristic functions using a Taylor expansion, since all moments are in $[0,1]$: 
\begin{align}
    \psi_p(t)=\sum_{\k=0}^\infty \frac{(\sqrt{-1}t)^\k}{\k!}\xmomi,\;\;\;
        \psi_q(t)=\sum_{\k=0}^\infty \frac{(\sqrt{-1}t)^\k}{\k!}\ymomi.
\end{align}

Then we can write 
\begin{align}
    &\int_0^T\frac{|\psi_q(t)-\psi_p(t)|}{t}dt\\
    &\leq\int_0^T\frac{\sum_{\k=1}^d \frac{t^\k}{\k!}|\xmomi-\ymomi|+\sum_{\k=d+1}^\infty\frac{t^{\k}}{\k!}}{t}dt\label{line:moments_less_than_1}\\
        &=\sum_{\k=1}^d \frac{T^\k}{\k!\k}\left|\xmomi-\ymomi\right| + \sum_{\k=d+1}^\infty\frac{T^{\k}}{\k!\k} \label{eq:abs_conv_int}\\
    &\leq\sum_{\k=1}^d \frac{T^\k}{\k!}\cdot \gamma+\sum_{\k=d+1}^\infty\frac{T^{\k}}{\k!\k}\label{eq:gamma_ratio_cond_Esseen}\\
      &\leq e^T\cdot \left(\gamma  +  \frac{T^{d+1}}{d!d}\right).\label{line:exp_T}
\end{align}
In~\eqref{line:moments_less_than_1} we used $|\xmomi-\ymomi|\leq 1$ for any $\k$ since $p,q$ have support in $[0,1]$; we used Fubini's Theorem to switch the order of integration and summation in~\eqref{eq:abs_conv_int} since the infinite sum is absolutely convergent; inequality~\eqref{eq:gamma_ratio_cond_Esseen} follows from the assumption on moment differences; and in~\eqref{line:exp_T} we used $\sum_{\k=1}^d\frac{T^\k}{\k!}<e^T-1$, together with $\sum_{\k=d+1}^\infty\frac{T^{\k}}{\k!\k} \leq \frac{T^de^T}{d!d}$. Plugging this in~\eqref{eq:esseen_fainleib}, the proof is concluded. 
\end{proof}

Next, in Lemmas~\ref{lem:upper_bound} and~\ref{lem:LowerBound}, we show that there exist a choice of $\bld{\beta},\bld{\gamma}$ such that $\Bvepsd(\v,\bld{\beta}) \subseteq \hist(\p,\grH,\delta)\subseteq \Bvepsd(\v,\bld{\gamma})$. 

\begin{remark}
Lemma~\ref{lem:KS_moments_equivalence} relates closeness in lower moments to closeness in KS distance. We note that similar results have been recently established for the $p$-Wasserstein distance~\cite{rigollet2019uncoupled} induced by the Euclidean metric on $\mathbb{R}$. The $1$-Wasserstein distance is of particular relevance to our setting, since it can be written as 
\begin{align}\label{eq:Wasser_def}
    \mathcal{W}_1(\p, q) = \int_{0}^1  \left|\int_0^x dp - \int_0^x dq\right| dx.
\end{align}
for any two probability distributions $\p$ and $q$ supported in $[0,1]$. Namely, $\mathcal{W}_1(\p, q)$ is the \textit{average} absolute distance between the c.d.fs., compared to $\KS(\p,q)$ which is the \textit{maximal} absolute distance. Here, we chose to work with the KS distance since it has a clear operational meaning, bounding the deviation in the number of vertices with an $\grH$-degree in any given interval, whereas working with the 1-Wasserstein distance would correspond to a weaker notion of (roughly) an average deviation over intervals. If one is nevertheless interested in working with the 1-Wasserstein distance, then~\cite{rigollet2019uncoupled} gives the following elegant bound for $X\sim p$ and $Y\sim q$:
\begin{align}\label{eq:rigollet}
   \frac{1}{m}|\E X^m-\E Y^m| \leq \mathcal{W}_1(\p, q)\leq \max \left\{\frac{1}{\log(1/\delta)},\frac{1}{d}\right\}. 
\end{align} 
The lower bound (stemming from the Kantorovich-Rubinstein duality) holds for any integer $m$, and the upper bound is with $\delta = \max_{m\in[d]}|\E X^m-\E Y^m|$. 
Using this relation we can approximate the $\delta$-Wasserstein ball around the reference distribution $\p$ by two $\F^d$-densities sets similar to these of Lemmas~\ref{lem:upper_bound} and \ref{lem:LowerBound}, following the same proof techniques. 
Note that it follows from~\eqref{eq:Wasser_def} that for measures over the unit interval
\begin{align}
    \mathcal{W}_1(\p, q)\leq \KS(\p, q), 
\end{align}
hence our upper bound \eqref{eq:KS_upper_bound_function_of_T} yields a bound on the 1-Wasserstein distance. Conversely, however, an upper bound on the 1-Wasserstein distance cannot in general yield an upper bound on the KS distance, in particular when one of the distributions is discrete as is the case in our setting. Hence, the bound~\eqref{eq:rigollet} from~\cite{rigollet2019uncoupled} cannot be used in our setting. 
This issue also manifests itself in the fact that our bound depends on $\p$ and $q$, whereas~\eqref{eq:rigollet} does not. 
\end{remark}

\begin{lem}\label{lem:upper_bound}
Let $\p$ be a distribution on $[0,1]$ that is absolutely continuous w.r.t. the Lebesgue measure, with a density lower bounded by some constant $a>0$. Let $d\geq 1$ be an integer and define $\v(p) \in [0,1]^d$ as in~\eqref{eq:density_approximation_out_of_p}. 
 Then
\begin{align}
    \hist(\p,\grH,\delta)\subseteq \Bvepsd(\v(p),\bld{\gamma})
\end{align}
for $\boldsymbol{\gamma}_\k=\frac{2\k \ci}{a} \E X^\k\cdot \delta+O(n^{-1})$, where $\ci$ is given in~\eqref{eq:asymptotic_normalization_ratio} and $X\sim p$. 
\end{lem}
\begin{proof}
Let $G\in \hist(\p,\grH,\delta)$. Then
\begin{align}
    \KS(\p, \pGF)\leq  \delta,
\end{align}
hence from Lemma~\ref{lem:KS_moments_equivalence}
\begin{align}
    |\xmomi-\ymomi|\leq \delta,
\end{align}
where $X\sim p$ and $Y\sim \pGF$.
Further note that 
\begin{align}
    \xmomi &\geq \int_{0}^1 a x^\k dx =\frac{a}{\k+1}.
\end{align}
Hence we get
\begin{align}
    \frac{|\xmomi-\ymomi|}{\xmomi}\leq\frac{\delta}{\xmomi}\leq \frac{(\k+1)}{a}\cdot \delta\leq \frac{2\k}{a}\cdot \delta. 
\end{align}
which yields
\begin{align}
 \frac{| \vm(\p)-\vm(\pGF)|}{\vm(p)}\leq \frac{2\k }{a}\cdot \delta,
\end{align}
and from Lemma~\ref{lem:density_equal_moment}
\begin{align}
    |t(G,\grH^\k)-\vm(p)| \leq  \frac{2\k }{a} \vm(p)\cdot \delta+O(n^{-1}), 
\end{align}
where we have used~\eqref{eq:density_approximation_out_of_p}. Therefore $G\in \Bvepsd(\v(p),\bld{\gamma})$. 
\end{proof}

\begin{lem}\label{lem:LowerBound}
Let $\p$ be a distribution over the unit interval, $d\geq 1$ an integer, and $\v(p)$ as in~\eqref{eq:density_approximation_out_of_p}. Set any $\bld{\beta}\in[0,1]^d$ satisfying 
\begin{align}
    \beta_\k \leq \ci \left(e^{-T}\left(\frac{\delta}{51} - S_p(1/T)\right) - \frac{T^{d+1}}{d!d}\right), 
\end{align}
if possible, where $\ci$ given in~\eqref{eq:asymptotic_normalization_ratio} and $T > 1$ arbitrary. 
Then 
\begin{align}
    \Bvepsd(\v(p), \bld{\beta}) \subseteq \hist(\p,\grH,\delta).
\end{align}
\end{lem}
\begin{proof}
Let $G\in \Bvepsd(\v(p), \bld{\beta})$. Then by definition
\begin{align}
    |t(G,\grH^\k)-\vm(\p)|\leq \beta_\k,
\end{align} 
and by Lemma~\ref{lem:density_equal_moment}
\begin{align}
    |\vm(\p)-\vm(\pGF)|\leq \beta_\k+O(n^{-1}).
\end{align} 
Using~\eqref{eq:density_approximation_out_of_p} yields 
\begin{align}
    |\xmomi-\ymomi|\leq \beta_\k / \ci +O(n^{-1}),
\end{align}
for $X\sim \p$ and $Y\sim \pGF$. Therefore, picking $\beta_\k$ as in the lemma, and appealing to Lemma~\ref{lem:density_equal_moment}, we get 
\begin{align}\label{eq:KS_smaller_delta}
\KS(\p,\pGF)\leq \delta,
\end{align}
concluding the proof.
\end{proof}


\subsection{Densities Set Size} \label{sec:sol_of_densities_problem}

From Lemmas~\ref{lem:upper_bound} and~\ref{lem:LowerBound} we conclude that 
$\Bvepsd(\boldsymbol{\phi}(p),\boldsymbol{\beta}) \subseteq  \hist(\p,\grH,\delta) \subseteq  \Bvepsd(\boldsymbol{\phi}(p),\boldsymbol{\gamma})$, namely, that the $\grH$-histogram is sandwiched between two densities sets with the same center vector $\phi$ and different radii. It is left to relate the size of the densities sets  to the maximum entropy problem~\eqref{opt:MP}. We will do this for a general family $\F$ of $d$ graphs (not necessarily our family $\F^d$).  
Before we proceed, recall that Chatterjee and Varadhan~\cite{chatterjee2011large} completely characterized the exponential growth of typical set in the case of a single density ($d=1$) in the limit of $n\rightarrow \infty$, as a solution to a maximum entropy problem over graphons. Furthermore, Lubetzky and Zhao~\cite{lubetzky2015replica} related the existence of this tight asymptotic characterization to the fact that a single density is a ``nice graph paramater'', which means that it is continuous in the cut metric with all local exterma being global extrema. Now, for $d>1$, this niceness property can be generalized by defining a local boundary point of the mapping $\tau(\S)$ as a point $\S$ in the  Szemer\'edi domain for which there exists a small enough $\delta$ such that there is no $\gamma>0$ for which the image of the ball $\{\S':\|\S-\S'\|_1\leq \delta\}$ in the parameter domain contains a ball of radius $\gamma>0$ around $\tau(S)$. Loosely speaking, local boundary point means that there is some direction that cannot be traversed in the parameter space, by small changes in the $\S$-space. A global boundary point is then defined as a point $\S$ in the  Szemer\'edi domain such that for any $\delta>0$ there is no $\gamma$ for which the image of the ball $\{\S':\|\S-\S'\|_1\leq \delta\}$ in the parameter domain contains a ball of radius $\gamma$ around $\tau(S)$. This, loosely speaking, means that  there is a direction that cannot be traversed at all, namely, it is really a boundary point of the image of the $\S$-space in the parameter space. 
Niceness therefore guarantees that for any two parameter vectors, if the straight line connecting them crosses no global boundary in the parameter domain, there exists a path in the graphon space (which in our finite setting is the Szemer\'edi type space) that maps to that straight line.
However, this property is not sufficient to facilitate a finite $n$ characterization; to that latter end, the magnitude of the gradient along this line cannot be too small, a property that will be captured here in our definition of the effective radius $\rho(S)$, which in turn must not be too large. 
Indeed, for a single density, it is possible to show that $\rho(S^*)$ vanishes as $\eps\to 0$ independently of the partition size $k$, hence we obtain an arbitrarily good accuracy in Theorem~\ref{thm:TypicalSetSizeGeneral}. This will be done in Section~\ref{sec:single_density_solution}.

From hereon we define 
\begin{align}
    \bar{r}\triangleq \max_{F\in\F}|V(F)|.
\end{align}
Let $\J(\S)$ be the Jacobian matrix of the function $\t(S,\F)$, which maps random graphs $\S$ on $k$ vertices to their vector of densities w.r.t. the family $\F$, as a function of $\S$. 
Note that since $\t(S,\F)$ are polynomials in $s_{ij}$, $\J(\S)$ always exist and is a $\binom{k}{2} \times \d$ matrix. Let 
\begin{align}
    \sigma(\rho,S) \triangleq \min_{\S':\|\S'-\S\|_1 \leq \rho} \sigma_{\min} \left(J(\S')\right)
\end{align}
be the smallest singular value of $J$ inside an $L^1$-ball of radius $\rho$ around $S$, where we set $\sigma_{\min}(J(S)) = 0$ for all $S'$ that lie outside this ball. 
We then define the effective radius $\rho(S)$ of $S$ to be the smallest radius $\rho$ for which 
\begin{align}\label{eq:g_parameter_of_S}
    \rho\cdot \sigma(\rho,S) \geq \twoctilde d \cdot \eps^{\frac{1}{\bar{r}}}.
\end{align}
We set $\rho(S) = \binom{k}{2}$ if no such radius exists. 
Loosely speaking, $\rho(S)$ is an upper bound on the deviation in the $S$-domain that guarantees a $\twoctilde d \cdot \eps^{\frac{1}{\bar{r}}}$ deviation in the densities domain.
Note that our definition of $\rho(\S)$ implicitly requires that for any densities vector in the neighborhood of $\tau(\S)$ there is some $\S'$ close to $\S$ that corresponds to the desired densities vector, thus implying the "local niceness" we mentioned in Subsection~\ref{sec:intro_CV_bounds}.

\begin{thm}\label{thm:TypicalSetSizeGeneral}
 Fix $\v, \bld{\gamma}\in[0,1]^d$ and $\eps<\frac{1}{\bar{r}^3}$ for $d\in\mathbb{N}$. Let 
\begin{align}
    S^* =& \argmax_\S H(\S) \label{eq:maximum_entropy_S}\\
    & \mathrm{s.t.} \quad\|\wS-\w\|_\infty \leq \bld{\gamma}+\ctilde \eps^{\frac{1}{\bar{r}}}
 \end{align} 
be any maximizer, where the maximization is over all $(k,\eps)$-Szemer\'edi types. The size of the typical set $\Bveps$ satisfies 
\begin{align}
   \left|\frac{1}{n^2}\log |\Bveps|- \frac{H(S^*)}{k^2}\right| \leq 5 h\left(\frac{\rho(S^*)}{4k^2}\right)+2\eps + o_\eps(1).
\end{align}
\end{thm}

Before we prove Theorem~\ref{thm:TypicalSetSizeGeneral}, we establish an important property of densities sets, which is that all the graphs in the same Szemer\'edi class have approximately the same subgraph densities, which  is equal to the mean densities $\t(\S,\F)$. Moreover, the deviation from this mean densities is bounded. This property is summarized in the following lemma whose proof follows standard counting arguments, see e.g. ~\cite{duke1995fast,rodl2005hypergraph,rodl2010regularity}.

\begin{lem}(counting lemma)\label{lem:counting_lemma}
Let $\grH$ be a graph on $r\geq 3$ vertices, and $\S$ be a $(k,\eps)$-Szemer\'edi type for $0<\eps<r^{-3}$. Then for any graph $G\in\Lambda(\S,n)$
\begin{equation}
    |t(G,\grH)-t(\S,\grH)|\leq \ctilde \eps^{\frac{1}{r-2}}.
\end{equation}
\end{lem}

\begin{proof}
\ifarxiv
See Appendix~\ref{appendix:counting_lemma}.
\else
Follows the same arguments as in similar previous variations of the counting lemma, e.g., \cite{duke1995fast,rodl2005hypergraph,rodl2010regularity}.
\fi
\end{proof}

We are now ready to prove Theorem~\ref{thm:TypicalSetSizeGeneral}. We use the following approach to evaluate the size of the $\bld{\gamma}$-typical set $\Bveps$: 
first, we define a new set, $\Cw$, which is the union of all the Szemer\'edi type classes $\Lambda(\S,n)$ for which $|\wSi-\vm|<  \gamma_\k+\ctilde \eps^{\frac{1}{\bar{r}}}$, $\k\in[d]$. Note that by Lemma~\ref{lem:counting_lemma}, $\Cw$ is guaranteed to contain all Szemer\'edi type classes that have an intersection with $\Bveps$, though it may include classes that have no intersection as well. Hence, $\Cw$ contains $\Bveps$. 
In Lemma~\ref{thm:tildSizeDensity}, we bound the size of $\Cw$ by the size of the largest Szemer\'edi type class in $\Cw$, times a polynomial factor of the number of types, which in turn yields an upper bound on the size of $\Bveps$. We then proceed to prove that the size of $\Bveps$ is also lower bounded by the size of $\Cw$ up to a polynomial factor, by showing that there exists a Szemer\'edi type class contained in $\Bveps$, whose size deviates from the maximal one in $\Cw$ by a term inversely proportional to $\rho(S^*)$.

Before we prove our main theorem, we need to characterize the number of different Szemer\'edi types, 
and the size of a Szemer\'edi type class.
For brevity, we implicitly assume 
that the subset $C_0$ in the partition is empty. This has a negligible effect on our results and can be easily accounted for. 

\begin{lem}\label{lem:number_of_Szemeredi_types}
The number of nonempty $(k,\eps,n)$-Szemer\'edi type classes
is upper bounded by $(\frac{n^2}{k^2}+1)^{k^2}$.
\end{lem}
\begin{proof}
The alphabet size of a $(k,\eps,n)$-Szemer\'edi type matrix $\S$ that corresponds to a nonempty class $\Lambda(\S,n)$  is $\frac{n^2}{k^2}+1$. 
\end{proof}

\begin{lem}\label{lem:S_size}
The size of any nonempty Szemer\'edi type class $\Lambda(k,\eps,\S,n)$ satisfies
\begin{equation*}
    \frac{H(\S) + o_\eps(1)}{k^2} \leq \frac{1}{n^2}\log|\Lambda(k,\eps,\S,n)| \leq \frac{H(\S)}{k^2}+2\eps.
\end{equation*}
\end{lem}
\begin{proof}
See Appendix~\ref{appendix:proof_of_8}.
\end{proof}

\begin{lem}\label{thm:tildSizeDensity}
Let $\F$ be a family of $d$ graphs, and $\v,\bld{\gamma}\in [0,1]^d$. Define the set 
\begin{align}\label{eq:Bset}
    \Cw=\hspace{-0.5cm}\bigcup_{{\begin{array}{cc}
     \quad \S:|\wS-\v|\leq \tilde{\bld{\gamma}}\\
\end{array}}}\hspace{-1.5cm}\Lambda(k,\eps,\S,n),
\end{align}
with $\tilde{\bld{\gamma}}_\k\triangleq \gamma_\k+\ctilde  \eps^{\frac{1}{ \bar{r}}}$.
Then
\begin{equation*}
   \frac{H(\S^*)+o_\eps(1)}{k^2} \leq \frac{1}{n^2}\log |\Cw| \leq  \frac{H(\S^*) + o_\eps(1)}{k^2} +2\eps, 
\end{equation*}
where
\begin{align}\label{eq:UpperBound_S}
    \S^* =& \argmax_{\S} H(\S) \quad \mathrm{ s.t. }\quad \Lambda(\S,n)\subseteq \Cw
\end{align}
\end{lem}

\begin{proof}
From Lemma~\ref{lem:S_size} and the fact that $\Lambda(\S^*,n)\subseteq \Cw$ we have
\begin{equation*}
    |\Cw| \geq |\Lambda(\S^*,n)|\geq  2^{\frac{n^2}{k^2} (H(\S^*)+o_\eps(1))}.
\end{equation*}
Also, from the union bound we get
\begin{equation*}
    |\Cw| \leq \hspace{-1.7cm}\sum_{{\begin{array}{cc}
     \quad S:\Lambda(\S,n)\subseteq \Cw\\
\end{array}}}\hspace{-1.7cm}|\Lambda(\S^*,n)|\leq \left(\frac{n^2}{k^2}+1\right)^{k^2}\cdot 2^{\frac{n^2}{k^2} H(\S^*)+n^2 2\eps},
\end{equation*}
which concludes the proof.
\end{proof}

Next, in Lemma~\ref{thm:SizeDensity}, we bound the size of the typical set $\Bveps$ using the size of the set $\Cw$. 
To that end, we first 
show in Lemma~\ref{lem:ratio_of_class_sizes_vs_L1_norm} that small perturbations in $\S$ result in bounded change to the  Szemer\'edi type class size $\Lambda(\S,n)$.

\begin{lem}[continuity]\label{lem:ratio_of_class_sizes_vs_L1_norm}
Let $\S_1,\S_2$ be two $(k,\eps,n)$-Szemer\'edi types with nonempty Szemer\'edi type classes $\Lambda(\S_1,n),\Lambda(\S_2,n)$.  Then, 
\begin{equation}\label{eq:sizeRatio}
 \left|\frac{1}{n^2}\log\frac{\left|\Lambda(\S_1,n)\right|}{\left|\Lambda(\S_2,n)\right|}\right| \leq 5 h\left(\frac{\|\S_1-\S_2\|_1}{4k^2}\right).
\end{equation}
\end{lem}
\begin{proof}
Denote $[\S_1]_{ij}\triangleq s_{ij}^{(1)}$, $[\S_2]_{ij}\triangleq s_{ij}^{(2)}$, $\delta_{ij}\triangleq |s^{(1)}_{ij}-s^{(2)}_{ij}|$ and assume without loss of generality $s^{(1)}_{ij}<s^{(2)}_{ij}$, and $s^{(1)}_{ij}<s^{(2)}_{ij}<\frac{1}{2}$. We will justify these assumptions in the end. 
First let us show that for any $0<\alpha<\frac{1}{2}$,
\begin{equation}
 \left|\log\frac{\left|\Lambda(\S_1,n)\right|}{\left|\Lambda(\S_2,n)\right|}\right| \leq n^2\cdot h(\alpha)+\frac{n^2}{k^2}\cdot \log\left(\frac{1-\alpha}{\alpha}\right)\cdot\|\S_1-\S_2\|_1.
\end{equation}
We treat three cases. First, assume $\alpha<s^{(1)}_{ij}<s^{(2)}_{ij}<1/2$, then we have
\begin{align}
    h(s^{(2)}_{ij})-h(s^{(1)}_{ij})&\leq \delta_{ij}\log\left(\frac{1-s^{(1)}_{ij}}{s^{(1)}_{ij}}\right)\label{line:convexity}\\
    &\leq \delta_{ij}\log\left(\frac{1-\alpha}{\alpha}\right)\label{line:alpha_smaller_than_s1}
    \leq h(\alpha)+ \delta_{ij}\log\left(\frac{1-\alpha}{\alpha}\right),
\end{align}
where \eqref{line:convexity} stems from the concavity of $h(x)$ and using the derivative of $h(\cdot)$, and \eqref{line:alpha_smaller_than_s1} is due to $\alpha<s^{(1)}_{ij}$ and the non-negativity of the binary entropy. 
Next, let $s^{(1)}_{ij}<\alpha<s^{(2)}_{ij}<1/2$, then we  have
\begin{align}
    h(s^{(2)}_{ij})-h(s^{(1)}_{ij})&\leq h(\alpha)+\left(s^{(2)}_{ij}-\alpha\right)\log\left(\frac{1-\alpha}{\alpha}\right)-h(s_{ij}^{(1)})\label{line:convexity_again}\\
    &\leq h(\alpha)+\left(\delta_{ij}-\alpha\right)\log\left(\frac{1-\alpha}{\alpha}\right)\label{line:s1_smaller_than_h}\\
    &\leq h(\alpha)+\delta_{ij}\log\left(\frac{1-\alpha}{\alpha}\right)\label{line:none_negative_log},
\end{align}
where in \eqref{line:convexity_again} we again used concavity of $h(x)$ and the derivative of $h(\cdot)$, in \eqref{line:s1_smaller_than_h} we used the relation $s\leq h(s)$ for $s\leq 1/2$ and in \eqref{line:none_negative_log} we used the fact that for $\alpha<1/2$ we  get $\log\left(\frac{1-\alpha}{\alpha}\right)>0$.
Finally, when  $s^{(1)}_{ij}<s^{(2)}_{ij}<\alpha\leq 1/2$, 
 \begin{equation*}
    h(s^{(2)}_{ij})-h(s^{(1)}_{ij})\leq h(\alpha)\leq h(\alpha)+\delta_{ij}\log\left(\frac{1-\alpha}{\alpha}\right).
\end{equation*}
Therefore, we get in the general case that
\begin{equation*}
    H(\S_2)- H(\S_1) \leq   k^2\cdot h(\alpha)+\log\left(\frac{1-\alpha}{\alpha}\right)\cdot\| \S_2-\S_1\|.
\end{equation*}
This result holds also when $s^{(1)}_{ij}<1/2<s^{(2)}_{ij}$ since in this case we can always take $\bar{s}^{(2)}_{ij}=1-s^{(2)}_{ij}$ and get the same entropy $h(s^{(2)}_{ij})=h(\bar{s}^{(2)}_{ij})$ but a smaller $\delta_{ij}$.
Next, let us choose $\alpha = \frac{\| \S_1-\S_2\| _1}{4k^2}$, then the above becomes
\begin{align}
    H(\S_2)- H(\S_1) &\leq k^2\left( h(\alpha)+4\alpha\log\left(\frac{1-\alpha}{\alpha}\right)\right)\\
    &\leq k^2\left( h(\alpha)-4\alpha\log\alpha +4\alpha\log(1-\alpha)\right)\\
    &\leq k^2\left( h(\alpha)-4\alpha\log\alpha\right)\\
    &\leq 5k^2h(\alpha).
\end{align}
and \eqref{eq:sizeRatio} then follows directly by plugging in the result of Lemma~\ref{lem:S_size}. 
\end{proof}

\begin{lem}\label{thm:SizeDensity}
Let $\v,\bld{\gamma}\in [0,1]^d$, $\F$ a family of $d$  graphs and $\S^*$ be as in $\eqref{eq:maximum_entropy_S}$. The size of the typical set $\Bveps$ satisfies 
\begin{equation}
\left| \frac{1}{n^2}\log \frac{|\Bveps|}{|\Cw|}\right| \leq 5h\left(\frac{\rho(S^*)}{4k^2}\right) + \frac{8k}{n}
\end{equation}
where $\Cw$  is defined as in  \eqref{eq:Bset}. 
\end{lem}

\begin{proof}
Recall that for any $G\in \Bveps$ we have $|\wGi -\vm|<\gamma_\k$. 
From  Lemma~\ref{lem:counting_lemma} we have that for any $G\in \Lambda(\S,n)$ it holds that $\|\wG-\wS\|_\infty\leq \ctilde  \eps^{\frac{1}{\bar{r}}}$. Then, using the triangle inequality we get that all the Szemer\'edi classes that contain some graph $G\in \Bveps$ must have $|\wSi-\vm|<\gamma_\k+\ctilde  \eps^{\frac{1}{\bar{r}}}$. Hence, $\Bveps\subseteq \Cw$ and $|\Bveps|\leq |\Cw|$.

Since $\S^*$ is the type with the largest type-class $\Lambda(\S^*,n)\subseteq \Cw$, we have $|\Cw| \leq (\frac{n^2}{k^2}+1)^{k^2}|\Lambda(\S^*,n)|$. For any (nonempty) $(k,\eps)$-Szemer\'edi type class $\Lambda(\S,n)\subseteq \Bveps$ we have $|\Lambda(\S,n)|\leq |\Bveps|$, and hence 
\begin{align}
    1\leq \frac{|\Cw|}{|\Bveps|} \leq \left(\frac{n^2}{k^2}+1\right)^{k^2}\frac{|\Lambda(\S^*,n)|}{|\Lambda(\S,n)|}.
\end{align}
Therefore, it is enough to show that there exists some nonempty $(k,\eps)$-Szemer\'edi type class $\Lambda(\S,n)\subseteq \Bveps$ with $ \left|\log\frac{\left|\Lambda(\S^*,n)\right|}{\left|\Lambda(\bar{\S},n)\right|}\right| \leq  5 n^2\cdot h\left(\frac{\rho(\S^*)}{4k^2}\right)$ to conclude our proof. 

Clearly, if $\Lambda(\S^*,n)\subseteq \Bveps$ we are done. 
Else, let $\hat{\w}_*=\frac{\w-\wSstar}{\|\w-\wSstar\|_2}$, i.e., the unit vector in the densities domain that points in the direction of $\w-\wSstar$. 
Note that for any $\S\in\Cw$ we have $|\wSi-\vm|<\gamma_\k+\ctilde  \eps^{\frac{1}{ r}}$, and a sufficient condition for $\Lambda(\S,n)$ to be fully contained in $\Bveps$ is $|\wSi-\vm|<\gamma_\k-\ctilde  \eps^{\frac{1}{\bar{r}}}$. 
Then, from the triangle inequality  we get that any $\S$ with $\wS=\wSstar+\twoctilde\d \cdot  \eps^{\frac{1}{\bar{r}}}\hat{\w}_*$ will ensure that $\Lambda(\S,n)\subseteq \Bveps$. 
From \eqref{eq:g_parameter_of_S} we get that there is a trajectory of length at most $\rho(S^*)$ in the $S$-domain that starts at $\S^*$ and ends at some $\bar{S}$ with $\w(\bar{\S})=\wSstar+\twoctilde\d  \cdot \eps^{\frac{1}{\bar{r}}}\hat{\w}_*$, and whose image in the densities domain is the straight line between $\wSstar$ and $\w(\bar{\S})$. 
Hence $\|\bar{\S}-S^*\|_1 \leq \rho(S^*)$. Here, we assume that $\bar{S}$ corresponds to a nonempty class $\Lambda(\bar{\S},n)$. In the case it is not, that is, the entries of $\bar{S}$ are not an integer multiple of the resolution $\frac{k^2}{n^2}$, we will need to quantize it to its nearest nonempty type. This is always possible since $\w$ is $d$-Lipschitz, hence for a large enough $n$ one can quantize $\bar{S}$ per entry as $\bar{s}_{ij,\text{quantized}}=\lceil \bar{s}_{ij}\cdot \frac{n^2}{k^2}\rceil \cdot \frac{k^2}{n^2}$ and still have $\bar{\S}_\text{quantized}\in\Bveps$.
Then, from Lemma~\ref{lem:ratio_of_class_sizes_vs_L1_norm} we get 
\begin{align}
    \log\frac{\left|\Lambda(\S^*,n)\right|}{\left|\Lambda(\bar{\S},n)\right|} \leq  5n^2h\left(\frac{\rho(S^*)}{4k^2}\right),
\end{align}
and along with $\log(\frac{n^2}{k^2}+1)^{k^2}\leq \frac{8k}{n}$ we are done.
\end{proof}
\begin{proof}[Proof of Theorem~\ref{thm:TypicalSetSizeGeneral}]
Follows from Lemmas~\ref{thm:tildSizeDensity} and~\ref{thm:SizeDensity}. 
\end{proof}

\section{The scalar $\grH$-density problem revisited}\label{sec:single_density_solution}

The accuracy of the estimated size of the densities set in Theorem~\ref{thm:TypicalSetSizeGeneral} is given as a function of $\rho(\S^*)$, which may yield a tight estimate for some $\S^*$ and not for others. 
However, in the special case $d=1$, that is, when the density-typical set is defined by a single subgraph density, an explicit expression for the accuracy can be derived, yielding a tight result for any $\S^*$, as we show below in Theorem~\ref{lem:single_density_typical_set_size}. Moreover, this case also vividly demonstrates the additional conditions needed in order to derive a finite-$n$ result, conditions which are not needed in the asymptotic case~\cite{chatterjee2011large}. Specifically, as shown in the proof of Theorem~\ref{lem:single_density_typical_set_size}, it is not enough that the parameter (in this case, the $\grH$-density) is nice; an additional bound on the size of its gradient is required. We establish such a bound in Lemma~\ref{lem:single_density_property}, where we show that the smallest deviation $\|\S-\bar{\S}\|$ in the $S$-domain that guarantees a deviation $\Delta$ in the density domain, is bounded by a quantity  proportional to $\Delta^{\sfrac{1}{\binom{r}{2}}}$.
We then use this result to get an (asymptotically) tight bound on the size of the density-typical set in Theorem~\ref{lem:single_density_typical_set_size}. 

\begin{lem}\label{lem:single_density_property}
Let $F$ be a graph on $r$ vertices and $\S$ be an $(k,\eps)$-Szemer\'edi type with $t(\S,\grH)=\phi$, then for any $\phi'\in[0,1]$ there exist a $(k,\eps)$-Szemer\'edi type $\bar{\S}$ with $t(\bar{\S},\grH)=\phi'$ and $\|\S-\bar{\S}\|_1 \leq \left(\frac{|\phi-\phi'|}{1-\min\{\phi,\phi'\}}\right)^{\frac{1}{\binom{r}{2}}} \binom{k}{2}$.
\end{lem}

\begin{proof}
To prove this we use a technique similar to the proof of \cite[Proposition~4.2]{chatterjee2011large}.
For simplicity we prove this lemma for cliques. The adjustment for general graphs $\grH$ is trivial since we consider non-induced graphs.
Assume without loss of generality that $\phi'\geq \phi$, and let  $\bar{\S}=\S+\alpha(1-\S)$ for some $\alpha\in[0,1]$ to be chosen later in the proof. 
Then for any $r$let $\mathcal{A}$ of indices we have
\begin{align*}
    \prod_{ij\in\mathcal{A}}\bar{s}_{ij}&=\prod_{ij\in\mathcal{A}}(s_{ij}+\alpha(1-s_{ij}))\\
    &\geq \prod_{ij\in\mathcal{A}}s_{ij}+\alpha^{\binom{r}{2}}\prod_{ij\in\mathcal{A}}(s_{ij}+(1-s_{ij}))- \alpha^{\binom{r}{2}}\prod_{ij\in\mathcal{A}}s_{ij} \\
    &= \prod_{ij\in\mathcal{A}}s_{ij}+\alpha^{\binom{r}{2}}(1-\prod_{ij\in\mathcal{A}}s_{ij}).
\end{align*}
Therefore  
\begin{align*}
    t(\bar{\S},\grH)\geq t(\S,\grH)+ \alpha^{\binom{r}{2}}(1-t(\S,\grH)) =\phi +\alpha^{\binom{r}{2}}(1-\phi). 
\end{align*}
and for $\alpha= \left(\frac{|\phi-\phi'|}{1-\phi}\right)^{\frac{1}{\binom{r}{2}}}$  we get
\begin{align*}
    t(\bar{\S},\grH)\geq \phi'.
\end{align*}
Since $t(\S,\grH)$ is continuous in all $s_{ij}$ we get by the intermediate value theorem that there exist some $\alpha\in[0,\left(\frac{|\phi-\phi'|}{1-\phi}\right)^{\frac{1}{\binom{r}{2}}}]$ for which $t(\bar{\S},\grH) =\phi'$.
Finally, note that 
\begin{align*}
    \|\S-\bar{S}\|_1=\alpha \sum_{ij} (1-s_{ij})\leq \alpha \binom{k}{2}, 
\end{align*}
 which concludes the proof.
\end{proof}

\begin{thm}\label{lem:single_density_typical_set_size}
 Fix $\phi, \gamma\in[0,1]$ and $\eps<\frac{1}{r^3}$. Let 
\begin{align}\label{eq:maximum_entropy_S_single_density}
    S^* =& \argmax_\S H(\S)\\
    & \mathrm{s.t.} \quad|t(\S,\grH)-\phi| \leq \gamma+\ctilde \eps^{\frac{1}{\bar{r}}}
 \end{align} 
be any maximizer, where the maximization is over all $(k,\eps)$-Szemer\'edi types. The size of the typical set $B_{\grH}^n(\phi,\gamma)$ satisfies 
\begin{align}
   \left|\frac{1}{n^2}\log |B_{\grH}^n(\phi,\gamma)|- \frac{H(S^*)}{k^2}\right| \leq 5 h\left(c\cdot \eps^{\frac{1}{r^3}}\right)+2\eps+o_\eps(1),
\end{align}
with \begin{align*}
    c=\begin{cases}\left(\frac{\twoctilde}{1-\phi-\gamma}\right)^{\frac{1}{\binom{r}{2}}},\;\phi+\gamma<1\\ 
    \left(\frac{\twoctilde}{1-\phi+\gamma}\right)^{\frac{1}{\binom{r}{2}}},\;\phi+\gamma\geq 1\end{cases}.
\end{align*}
\end{thm}
\begin{proof}
First, assume $\phi+\gamma<1$. 
It is enough to show that when $d=1$ the result in Lemma~\ref{thm:SizeDensity} can be improved to
\begin{equation}
\left| \frac{1}{n^2}\log \frac{|B_{\grH}^n(\phi,\gamma)|}{|C_{\grH}^n(\phi,\gamma)|}\right| \leq 5 h\left( c\cdot \eps^{\frac{1}{r^3}}\right) + \frac{8k}{n},
\end{equation}
where $C_{\grH}^n(\phi,\gamma)$ is as in \eqref{eq:Bset}. Then the rest of the proof follows easily using the same technique as in Theorem~\ref{thm:TypicalSetSizeGeneral}.
To that end, assume without loss of generality that $\phi+\gamma-\ctilde \eps^{\frac{1}{\bar{r}}}\leq t(\S^*,\grH)\leq \phi+\gamma+\ctilde \eps^{\frac{1}{\bar{r}}}$. 
We need to show that there exists a type $\bar{\S}$ with $t(\bar{\S},\grH)=t(\S^*,\grH) -\twoctilde \eps^{\frac{1}{\bar{r}}}$ to ensure that $\Lambda(\bar{S})\subseteq B_{\grH}^n(\phi,\gamma)$. 
From Lemma~\ref{lem:single_density_property} we get that there exist such $\bar{\S}$ with $\|\S^*-\bar{\S}\|\leq (\frac{\twoctilde}{1-\phi-\gamma })^{\frac{1}{\binom{r}{2}}}\cdot \eps^{\frac{1}{r^3}}$, which concludes the proof for the case $\phi+\gamma<1$.
The adjustment for the case $\phi+\gamma\geq 1$ can be easily done by noting that from the monotonicity of $t(\S,F)$ in all $s_{ij}$ it follows that for any $\eps<\frac{1}{r^3}$, the maximum entropy Szemer\'edi type \eqref{eq:maximum_entropy_S_single_density} must hold that $t(\S,F)<\phi+\gamma-\ctilde \eps^{\frac{1}{\bar{r}}}$, hence the constant in the theorem will be given by $c=\left(\frac{\twoctilde}{1-\phi+\gamma}\right)^{\frac{1}{\binom{r}{2}}}$.
\end{proof}

\section{Summary and Discussion}\label{sec:discussion}

In this paper we considered the problem of counting the number of  graphs on $n$ vertices that share approximately the same $F$-degree distribution.
Except for the special case when $F$ is a single edge, this problem had not been addressed before, and generalizing the methods used in the edge case appears to be nontrivial, partly due to the feasibility problem. 
Here, to circumvent the feasibility problem, we defined a histogram as a KS-ball around a smooth reference distribution, and then characterize the number of graphs whose $F$-degree distribution lies inside this ball, in terms of a solution to a constrained maximum entropy problem over fixed-dimension random graphs with global structure constraints. Our approach was based on reducing the problem to the study of multiple global density types, and then estimating the size of such types using the regularity lemma and anti-concentration inequalities. 

The main gap in the current work is deriving explicit continuity conditions for the maximum entropy solution \eqref{eq:maximum_entropy_S} in Theorem~\ref{thm:TypicalSetSizeGeneral}, when $d>1$.
To that end, one approach is bounding the smallest singular value of the Jacobian $J(\S)$ of the mapping $\t(\S,\F)$ away from zero. However, this cannot be done for all the feasible $\S$ points. For example, whenever $s_{ij}=s$ for all $i,j\in[k]$, and some $s\in[0,1]$, $J(\S)$ has at least one zero singular value. Hence, it seems that a more complex argument is needed in order to establish the desired continuity. 

Other interesting aspects for further study may include improving the exponential bound \eqref{eq:KS_upper_bound_function_of_T} to yield tighter upper and lower bounds in Theorem~\ref{thm:main_theorem}, and extending the framework to handle induced subgraphs. The fact that we count non-induced subgraphs plays a central role in both steps of our solution. When replacing the $F$-degree distribution constraint with $d$ global density constraint in Subsection~\ref{sec:reduction_to_densities}, we use the fact that the subgraphs we count are non-induced in order to establish the equivalence between the moments of the $F$-degree distribution and the global densities of the graph. Then, when characterizing the size of the density typical set in Subsection~\ref{sec:sol_of_densities_problem}, we assume that the (expected) global subgraph densities of the Szemer\'edi type are monotonically increasing in the edge densities $s_{ij}$. This assumption is no longer true when dealing with induced subgraphs.

\section{Acknowledgements}
We are indebted to Wojciech Samotij for useful discussions and many ideas that were elemental in writing this paper.
We thank the Associate Editor for his dedicated effort in handling our paper, and the reviewers for making many helpful and constructive comments that improved  our presentation. In particular, we are grateful to one of the reviewers for introducing us to~\cite{rigollet2019uncoupled}.



\appendix 
\section{Appendix}

\subsection{The CV Bounds}\label{appendix:CV_bounds}

We provide a brief formal presentation of the main large deviation result by Chatterjee and Vardhan in \cite{chatterjee2011large}, and then give a toy example showing that these bounds are not necessarily tight. 
Let $\mathcal{W}$ be the space of all graphons, i.e., of all measurable functions from $[0, 1]^2$ into $[0, 1]$ that satisfy $f (x, y) = f (y, x)$ for all $x$, $y$. Two elements $f,g\in\mathcal{W}$ are said to be equivalent if there exist a measure preserving bijection $\sigma:[0,1] \rightarrow [0,1]$ such that $f(x,y)=g(\sigma x,\sigma y)\triangleq g_\sigma(x,y)$. The quotient space induced by this equivalence relation is denoted by $\tilde{\mathcal{W}}$. The cut distance between two elements $f,g\in \mathcal{W}$ is given by 
\begin{align}
    d_\square(f,g)\triangleq \sup_{S,T\subseteq[0,1]}\int_{S\times T}[f(x,y)-g(x,y)]dxdy. 
\end{align}
The cut metric is then defined for two graphon $\tilde{f},\tilde{g}$ in $\tilde{\mathcal{W}}$ as
\begin{align}
    \delta_\square (\tilde{f},\tilde{g})= \inf_\sigma d_\square(f,g_\sigma).
\end{align}
It was shown in \cite{lovasz2012large} that $\tilde{\mathcal{W}}$ is compact with respect to the cut metric. 

A graph $G$ of $n$ vertices has a natural representation $f^G$ in the graphon space, given by \eqref{eq:fG}. 
The $\grH$-density of a graph as given in \eqref{eq:density_deterministic} also has a natural extension to graphons, given by
\begin{align}
    t(f,\grH) \triangleq \int_{[0,1]^r} \prod_{i\sim_\grH j}f(x_i,x_j)dx_1\cdots dx_r,
\end{align}
where $i\sim_\grH j$ indicates that there is an edge between the vertices $i$ and $j$ in $\grH$.
If $G_n$ is a sequence of simple graphs whose number of nodes tends to infinity, and $t(\grH,G_n)$ is the corresponding $\grH$-density of $G_n$, 
and the sequence $t(\grH,G_n)$, $n\rightarrow \infty$, converges for any $\grH$,  
then there exists  a graphon $f\in\mathcal{W}$, such that 
\begin{align}
    \delta_\square(\tilde{f}^{G_n},\tilde{f})\underset{n\rightarrow \infty}{\longrightarrow} 0.
\end{align}

Denote by $\mathbb{P}_n$ the probability induced on the space $\tilde{\mathcal{W}}$ by the Erd\H{o}s-R\'enyi random graph $G(n,1/2)$ through the map $G\rightarrow f^G\rightarrow \tilde{f}^G$ (the original result in \cite{chatterjee2011large} is stated for a general Erd\H{o}s-R\'enyi random graph $G(p,1/2)$, $p\in[0,1]$. We bring the special case $p=1/2$ here since this is the useful setting for the purpose of our counting problem).

\begin{thm}[Theorem 2.3 in \cite{chatterjee2011large}]\label{thm:CV_bounds}
For any closed set $\tilde{V}\subseteq \tilde{\mathcal{W}}$
\begin{align}
    \limsup_{n\rightarrow \infty} \frac{1}{n^2} \log\mathbb{P}_n(\tilde{V})\leq \sup_{f\in\tilde{V}}(H(f)),\label{eq:CV_upper}
\end{align}
and for any open set $\tilde{U}\in\tilde{\mathcal{W}}$,
\begin{align}
    \liminf_{n\rightarrow \infty} \frac{1}{n^2} \log\mathbb{P}_n(\tilde{U})\geq \sup_{f\in\tilde{U}}(H(f)).\label{eq:CV_lower}
\end{align}
\end{thm}

The above can be used to bound the number of graphs $G$ whose $\grH$-density lies within some interval $[a,b]$ by setting $\tilde{U}$ to be the open set $\tilde{U}=\{f\in\tilde{\mathcal{W}}: a < t(f,\grH)< b\}$ and $\tilde{V}$ as its closure $\tilde{V}=\{f\in\tilde{\mathcal{W}}: a\leq t(f,\grH)\leq b\}$.
Then by noting that $\mathbb{P}_n$ is uniform over all graphons that correspond to some graph $G$ on $n$ vertices and zero over all others, the desired result is obtained. 
A similar result can be derived for any graph parameter. The tightness of the bounds then depends, as explained in Subsection~\ref{sec:intro_CV_bounds}, on the continuity of the rate function $\psi(t)$ in \eqref{eq:rate_func} with respect to $t$.

To show that the bounds may not coincide in some cases, consider the following simple example;  let $e(f)$ denote the edge density of $f$, that is, $t(\grH,f)$ in the case that $\grH$ is a single edge. We define the parameter $\tau(f)$ to be the function $z(x)$ in Figure~\ref{fig:toy_example}, applied to $e(f)$, i.e., $\tau(f)= z(e(f))$. 
Since $e(f)$ is continuous with respect to the cut metric, and $\tau(f)$ is a continuous mapping of $e(f)$, then $\tau(f)$ is also continuous with respect to the cut metric and hence a valid graph parameter.

\begin{figure}[h]
\includegraphics[width=0.85\textwidth]{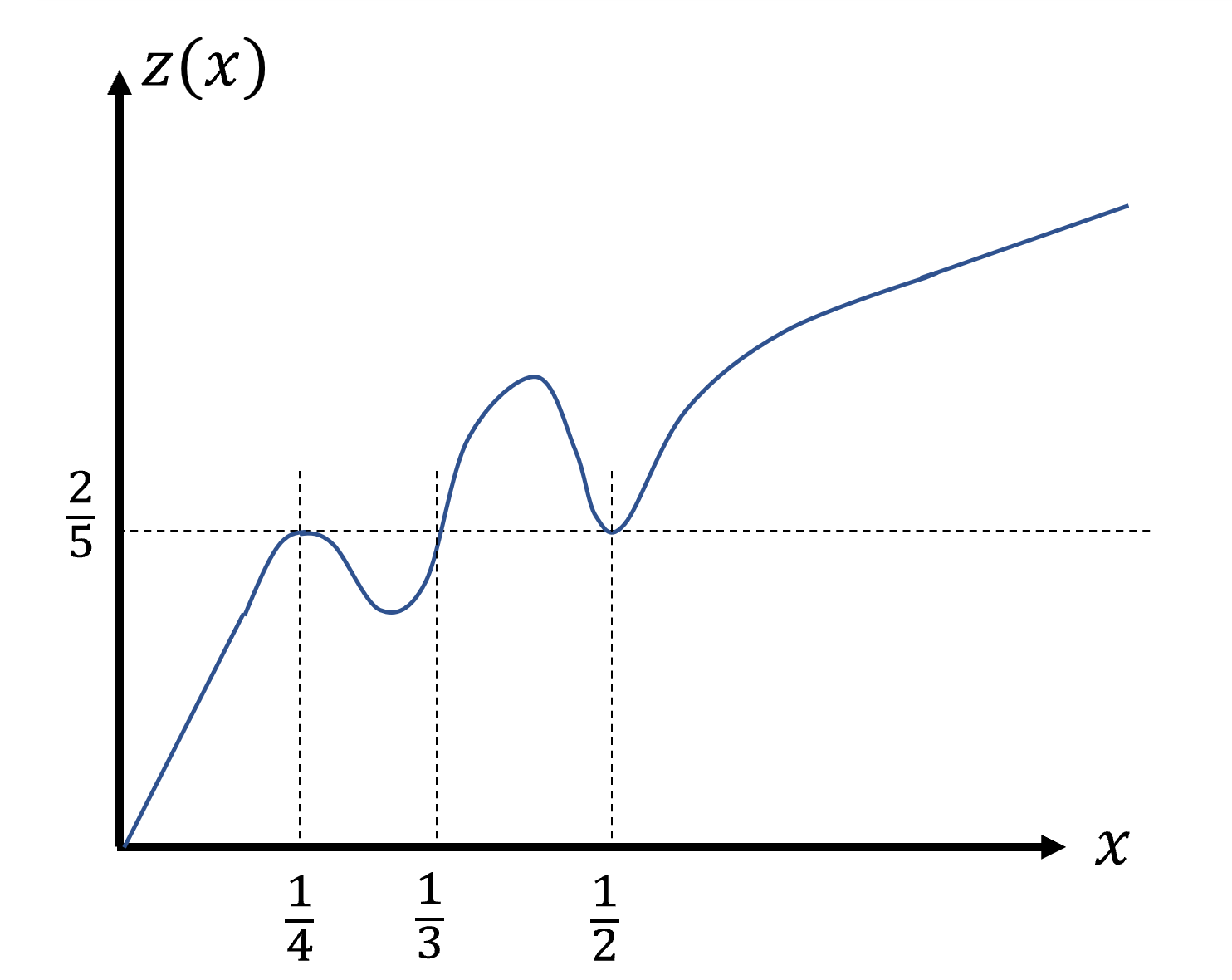}\label{fig:toy_example}
\caption{Example of a parameter for which the CV bounds do not coincide. Setting $\tau(f)=Z(e(f))$, where $e(f)$ is the edge density of the graphon $f$, will result in a gap between the upper bound (80) and lower bound (81), for the set ${f:\tau(f)<2/5}$ and its closure. This gap occurs since $\tau(f)$ is not a nice graph parameter; it has local extrema points that are not global ones.}
\end{figure}

Next, let us  apply the bounds in Theorem~\ref{thm:CV_bounds} to an open set and its closure, and see that they do not coincide (in some cases). Consider the sets $\tilde{U}=\{f:\tau(f)<\frac{2}{5}\}$ and $\tilde{V}=\{f:\tau(f)\leq \frac{2}{5}\}$. Then we get that the upper bound \eqref{eq:CV_upper} is equal to $1$, since $\tilde{V}$ contains the constant graphon $f(x,y)=1/2$, $\forall x,y$, which is the global maximum entropy graphon.  
However, the lower bound \eqref{eq:CV_lower} must be strictly smaller than $1$, since $\tilde{U}$ contains only graphons with $e(f)<1/3$. 
An example with a countable number of jump discontinuities can be constructed in a similar way.

\ifarxiv
\subsection{Proof of Szemer\'edi's Regularity Lemma }\label{appendix:Szemeredi_reg_lemma}
We prove Szemer\'edi's  regularity lemma, to make the paper more self-contained for readers less familiar with it. We prove a slightly weaker version of Lemma~\ref{lem:SzemerLem}, where the size of the sets in the partition are not necessarily equal sized. Also, we only show that for any $\eps$, for each graph there exist some $k$ such that it has an $(\eps,k)$ partition, rather that there is a single $k$ that fits all graphs (for $n$ large enough). We chose to prove this version, which is based on the one available on Wikipedia, since its proof is short and intuitive; the proof of the full version appearing in Lemma~\ref{lem:SzemerLem} can be found in \cite{szemeredi1975regular}. The course of the proof will be as follows: for a given graph we first start with some arbitrary partition of size $k_0$. Then, if this partition is not $\eps$-uniform we perform a refinement step where each set is partitioned to at most $2^{k_0}$ parts. We repeat this step while the partition is not $\eps$-uniform.
We will show that as long as the partition is not $\eps$-uniform, the refinement step increases the "energy" of the partition (to be defined later). Since this energy property is bounded from above, this process is finite and bound to produce some $\eps$-uniform partition.

To that end we need the following lemmas that will show the energy cannot decrease upon refinement, and most increase if the partition to be refined is not $\eps$-uniform.
Let $G$ be some graph on $n$ vertices and let $W, U\subseteq V(G)$. Define
\begin{align*}
    q(W,U)\triangleq \frac{|W||U|}{n^2}d(W,U)^2.
\end{align*}
For partitions $\mathcal{P}_W=\{W_1,\cdots,W_k\}$ of $W$ and $\mathcal{P}_U=\{U_1,\cdots,U_k\}$ of $U$ define
\begin{align*}
    q(\mathcal{P}_W,\mathcal{P}_U)\triangleq \sum_{i=1}^k \sum_{j=1}^k q(W_i,U_j).
\end{align*}
Then, for a partition $\mathcal{P}=\{C_1,\cdots, c_k\}$ of $V(G)$ we define the {\em energy} of the partition as 
\begin{align*}
    q(\mathcal{P})= \sum_{i=1}^k \sum_{j=1}^k q(C_i,C_j) = \sum_{i=1}^k \sum_{j=1}^k \frac{|C_i||C_j|}{n^2}d(C_i,C_j)^2.
\end{align*}
Note that $0\leq q(\mathcal{P}) \leq 1$ for any $\mathcal{P}$ since $0\leq d(C_i,C_j)\leq 1$ by definition.
First we show that the energy is non increasing upon refinement.
\begin{lem}\label{lem:non_increasing_eneregy}
Let $\mathcal{P}_W$ and $\mathcal{P}_U$ be some partitions of $W$ and $U$ respectively, then
\begin{align}
    q(\mathcal{P}_W,\mathcal{P}_U)\geq q(W,U).
\end{align}
\end{lem}
\begin{proof}
Let $\mathcal{P}_W=\{W_1,\cdots,W_k\}$ of $W$ and $\mathcal{P}_U=\{U_1,\cdots,U_k\}$, and let us choose a vertex $x$ from $W$ and a vertex $y$ from $U$ uniformly at random.
Let $W_i$ and $U_j$ be the subsets that $x$ and $y$ belongs to in the partitions $\mathcal{P}_W$ and $\mathcal{P}_U$, respectively, and define the random variable $Z=d(W_i,U_j)$.
Then
\begin{align}
    \mathbb{E}\left[Z\right]=\sum_{i=1}^k \sum_{j=1}^k \frac{W_i}{W}\frac{U_j}{U}d(W_i,U_j)=\frac{e(W,U)}{|W||U|}=d(W,U),
\end{align}
and
\begin{align}
    \mathbb{E}\left[Z^2\right]=\sum_{i=1}^k \sum_{j=1}^k \frac{W_i}{W}\frac{U_j}{U}d(W_i,U_j)^2=\frac{n^2}{|W||U|}q(\mathcal{P}_W,\mathcal{P}_U).
\end{align}
By convexity we have $\mathbb{E}\left[Z^2\right]\geq \mathbb{E}\left[Z\right]^2$ which yields the desired result.
\end{proof}
Next, we show that if a pair of sets is not $\eps$-uniform, there exist a refinement for the pair that will boost its energy.
\begin{lem}\label{lem:energy_boost}
If the pair of sets $(W,U)$ is not $\eps$-uniform as witnessed by $W_1\subset W$ and $U_1\subset U$, then 
\begin{align}
    q(\{W,W\backslash W_1\},\{U,U\backslash U_1\})>q(W,U)+\eps^4\frac{|W||U|}{n^2}.
\end{align}
\end{lem}
\begin{proof}
Define $Z$ as in the previous lemma. Then
\begin{align}
    \mathrm{var}\left[Z\right] = \mathbb{E}\left[Z^2\right]- \mathbb{E}\left[Z\right]^2= \frac{n^2}{|W||U|}q(\{W,W\backslash W_1\},\{U,U\backslash U_1\})-q(W,U).
\end{align}
Note that $|Z-\mathbb{E}[Z]|=|d(W_1,U_1)-d(W,U)|$ with probability $\frac{W_1}{W}\frac{U_1}{U}$, hence 
\begin{align}
     \mathrm{var}\left[Z\right] = \mathbb{E}\left[\left(Z-\mathbb{E}\left[Z\right]\right)^2\right]\geq \frac{W_1}{W}\frac{U_1}{U}(d(W_1,U_1)-d(W,U))^2>\eps ^4,
\end{align}
which concludes the proof.
\end{proof}

\begin{lem}\label{lem:partition_increase_energy}
If a partition $\mathcal{P} = {C_1,\cdots , C_k}$ of $V(G)$ is not $\eps$-uniform, then there exists a refinement $\mathcal{P}'$ of $\mathcal{P}$ 
in which every set $C_i$ is partitioned into at most $2^k$ parts and 
\begin{align}
    q(\mathcal{P}')\geq q(\mathcal{P})+\eps^5.
\end{align}
\end{lem}

\begin{proof}
For any $(i,j)$ such that $(C_{i},C_{j})$ is not $\eps$-uniform, find the subsets $A_{ij}\subseteq C_i$ and $A_{ji}\subseteq C_j$ that witness the irregularity. 
Let $\mathcal{P}'$ be the refinement of $\mathcal{P}$ by all the subsets $A_{ji}$, $1\leq,i,j\leq k$. Then, in this partition each set $C_i$ is partitioned into at most $2^{k}$ parts. 
Denote by $\mathcal{P}'_{C_i}$ the partitioning of $C_i$ in $\mathcal{P}'$, then,
\begin{align}
    q(\mathcal{P}')=\hspace{-0.5cm}\sum_{(i,j)\in [k]\times[k]}\hspace{-0.3cm}q(\mathcal{P}'_{C_i},\mathcal{P}'_{C_j})=\hspace{-0.8cm}\sum_{\tiny{\begin{array}{cc}(i,j)\in [k]\times[k]:\\(C_i,C_j)\text{ $\eps$-uniform}\end{array}}}\hspace{-0.8cm}q(\mathcal{P}'_{C_i},\mathcal{P}'_{C_j})+\hspace{-0.8cm}\sum_{\tiny{\begin{array}{cc}(i,j)\in [k]\times[k]:\\(C_i,C_j)\text{ not $\eps$-uniform}\end{array}}}\hspace{-0.8cm}q(\mathcal{P}'_{C_i},\mathcal{P}'_{C_j}).
\end{align}
Since $\mathcal{P}'_{C_i}$ is a refinement of $\{C_i, C_i\backslash A_{ij}\}$ we get  that 
\begin{align}
    q(\mathcal{P}')&\geq \hspace{-0.1cm}\sum_{\tiny{\begin{array}{cc}(i,j)\in [k]\times[k]:\\(C_i,C_j)\text{ $\eps$-uniform}\end{array}}}\hspace{-0.8cm}q(C_i,C_j)+\hspace{-0.8cm}\sum_{\tiny{\begin{array}{cc}(i,j)\in [k]\times[k]:\\(C_i,C_j)\text{ not $\eps$-uniform}\end{array}}}\hspace{-0.8cm}q(\{C_i,C_i\backslash A_{ij}\},\{C_j,C_j\backslash A_{ji}\})\label{line:non_inc} \\
   &\geq \sum_{(i,j)\in [k]\times[k]}q(C_i,C_j)+\hspace{-0.8cm}\sum_{\tiny{\begin{array}{cc}(i,j)\in [k]\times[k]:\\(C_i,C_j)\text{ not $\eps$-uniform}\end{array}}}\hspace{-0.8cm}\eps^4\frac{|C_i||C_j|}{n^2})\label{line:boost} \\
    &\geq q(\mathcal{P})+\eps^5,\label{line:not_eps_unif}
\end{align}
where \eqref{line:non_inc} is due to Lemma~\ref{lem:non_increasing_eneregy}, \eqref{line:boost} is due to Lemma~\ref{lem:energy_boost}, and \eqref{line:not_eps_unif} is since $\mathcal{P}$ is not $\eps$-uniform.
\end{proof}
\begin{proof}[Proof of Szemer\'edi regularity lemma]
We start with a trivial partition ($k=1$) and while the partition is not $\eps$-uniform we apply Lemma~\ref{lem:partition_increase_energy}.
At each step the energy of the partition increases by at least $\eps^5$, but $q(\mathcal{P})\leq 1$, hence we bound to stop, i.e., get an $\eps$-uniform partition, after at most $\eps^{-5}$ steps.
\end{proof}
\subsection{Proof of Fainleib's inequality}\label{appendix:Fainleib}
We now bring a sketch of the proof of Fainleib's inequality \eqref{eq:esseen_fainleib}. The full version including the exact value of the constants can be found in~\cite{fainleib1968generalization}.

\begin{proof}
The concept of the proof is to show that there exist a low pass filter $\varphi_T(x)$ such that for any cumulative distribution function $Q(x)$ it hold that $Q(x)\leq  2\int_x^{x+\frac{c_2}{T}}Q(u)\varphi_T(u-x)du$, for some constant $c$. Therefor for any $x$ and any two c.d.fs $Q(x)$ and $F(x)$ the expression $Q(x)-F(x)\leq  2\int_x^{x+\frac{c_2}{T}}Q(u)\varphi_T(u-x)du-F(x)\leq 2\int_x^{x+\frac{c_2}{T}}(Q(u)-F(u))\varphi_T(u-x)du+f(F)$ where the latter is a function that depend on the concentration of $F$.
Let $\varphi(x)$ be a filter such that for some constants $c_1,c_2,c_3$ the following properties hold:
\begin{enumerate}
    \item $0\leq \varphi_T(x)\leq c_1 T$, for all $x$; \label{propert1}
    \item $\int_{-\infty}^\infty \varphi_T(x)dx =1$;\label{propert2}
    \item $\int_0^{\frac{c_2}{T}}\varphi_T(x)dx=c(T)\geq \frac{1}{2}$; \label{propert3}
    \item $\vartheta(t)\triangleq \int_{-\infty}^\infty\varphi_T(x)e^{itx}dx$, the Fourier transform of $\varphi_T(x)$ is \label{propert4}
    \begin{itemize}
        \item band limited, $\vartheta(t)=0$, for all $t:|t|\geq T$,
        \item bounded, $|\vartheta(t)|\leq c_3$,
        \item symmetric in absolute value $|\vartheta(t)|=|\vartheta(-t)|$.
    \end{itemize}
\end{enumerate}
Then, since $Q(x)$ is a c.d.f, i.e. monotonically increasing function, and by letting $h=\frac{c_2}{T}$, we get that indeed 
\begin{align}
    Q(x)\leq  \frac{1}{c(T)}\int_x^{x+\frac{c_2}{T}}Q(u)\varphi_T(u-x)du,
\end{align}
where we used property~\ref{propert3}.
Then we can write
\begin{align}
    Q(x&)-F(x)\leq \frac{1}{c(T)}\int_x^{x+\frac{c_2}{T}}Q(u)\varphi_T(u-x)du - F(x)\\
    & = \frac{1}{c(T)}\int_x^{x+\frac{c_2}{T}}\hspace{-0.5cm}(F(u) - F(x))\varphi_T(u-x)du  + \frac{1}{c(T)}\int_x^{x+\frac{c_2}{T}}\hspace{-0.5cm}(Q(u)-F(u))\varphi_T(u-x)du\label{line:used_property_2}\\
    &\leq \frac{c_1 T}{c(T)}\int_0^{\frac{c_2}{T}}(F(x+u) - F(x))du+ \frac{1}{c(T)}\int_{-\infty}^{\infty}(Q(u)-F(u))\varphi_T(u-x)du  \nonumber\\
    &\qquad - \frac{1}{c(T)}\left(\int_{-\infty}^{x}+\int_{x+\frac{c_2}{T}}^{\infty}\right)(Q(u)-F(u))\varphi_T(u-x)du \label{line:used_property_1} \\
    &\leq \frac{c_1 T}{c(T)}\int_0^{\frac{c_2}{T}}(F(x+u) - F(x))du+ \frac{1}{c(T)}\int_{-\infty}^{\infty}(Q(u)-F(u))\varphi_T(u-x)du \nonumber \\
    &\qquad -\frac{1}{c(T)}\left(\int_{-\infty}^{X}+\int_{x+\frac{c_2}{T}}^{\infty}\right)(Q(u)-F(u))\varphi_T(u-x)du
    \label{line:used_property_2_again}.
\end{align}
where in \ref{line:used_property_2} and~\ref{line:used_property_2_again}  we used property~\ref{propert2} and in \ref{line:used_property_1} we used property~\ref{propert1}. 
Then, we have 
\begin{align}
    \frac{c_1 T}{c(T)}\int_0^{\frac{c_2}{T}}(F(x+u&) - F(x))du\leq 2c_1 T\int_0^{\frac{c_2}{T}}(F(x+u) - F(x-u))du \\
    &= c_1c_2\left(\frac{T}{c_2}\int_0^{\frac{c_2}{T}}(F(x+u) - F(x-u))du\right)\\
    &\leq c_1c_2\left(\frac{T}{c_2}\int_0^{\frac{c_2}{T}}(F(x+u) - F(x-u))du\right)\\
    &\triangleq  c_1c_2 \tilde{S}_f(\frac{c_2}{T}).
\end{align}
It can be easily shown that there exist a constant $c_4$ such that $\tilde{S}_f(\frac{c_2}{T})\leq c_4 \tilde{S}_f(\frac{1}{T})$ (full derivation in \cite{fainleib1968generalization}). Then \begin{align}
    \tilde{S}_f(\frac{1}{T})&=T\int_0^{\frac{1}{T}}(F(x+u) - F(x-u))du \leq \sup_x (F(x+u) - F(x-u)))\\
    &= S_f(\frac{1}{T}).
\end{align}
When $Q(x)$ and $F(x)$ corresponds to some continuous pdfs $q(x)$ and $f(x)$ respectively, then the characteristics functions $\psi_q$ and $\psi_f$ are their Fourier transforms and we obtain 
\begin{align}
    \left|\int_{-\infty}^{\infty}(Q(u)-F(u))\varphi_T(u-x)du\right| &= \left|\int_{-\infty}^\infty e^{-ixt} \left(\frac{\psi_q(t)}{t}-\frac{\psi_f(t)}{t}\right)\vartheta(t)dt\right|\\
    &\leq c_3\int_{-\infty}^\infty \frac{|\psi_q(t)-\psi_f(t)|}{t}dt,\label{line:used_property_4}
\end{align}
where in \eqref{line:used_property_4} we used the triangle inequality and property~\ref{propert4}.
When $Q(x)$ and $F(x)$ does not corresponds to continuous pdfs the same result can be obtained using Levy's inversion, full derivation available at \cite{fainleib1968generalization}. 
It is left then to deal with the term $\frac{1}{c(T)}\left(\int_{-\infty}^{X}+\int_{x+\frac{c_2}{T}}^{\infty}\right)(Q(u)-F(u))\varphi_T(u-x)du$.
Let $a=\sup_x |Q(x)-F(x)|$, then 
\begin{align}
    \Big|\Big(\int_{-\infty}^{X}+\int_{x+\frac{c_2}{T}}^{\infty}&\Big)(Q(u)-F(u))\varphi_T(u-x)du\Big|\leq a\left(\int_{-\infty}^{X}+\int_{x+\frac{c_2}{T}}^{\infty}\right)\varphi_T(u-x)du\\
    &=a\left(1-\int_x^{x+\frac{c_2}{T}}\varphi_T(u-x)du\right)=a(1-c(T)).
\end{align}
Then by combining all three terms together and switching wings we get the desired results. 
In \cite{fainleib1968generalization} it is shown that the filter 
\begin{align}
    \varphi_T(x)= \frac{T}{\pi} \frac{1-\cos(Tx-3)}{(Tx-3)^2},
\end{align}
hold all the required properties. For the exact constants the reader is referred to the full proof there.
\end{proof}

\subsection{Proof of the Counting Lemma}\label{appendix:counting_lemma}
\begin{proof}[Proof of Lemma~\ref{lem:counting_lemma}]
Here we prove the lemma for the case where $\grH$ is a clique. The proof can be easily amended to account for general graphs. In the following, an {\em irregular} pair is a pair of subsets in the $(k,\eps)$-uniform partition that is not an $\eps$-uniform pair. Also, we denote $g=\frac{k}{n}$.
Let $G\in \Lambda(\S,n)$ be a graph that corresponds to an $(k,\eps)$-uniform partition $G=C_0\bigcup C_1 \bigcup \cdots \bigcup C_k$. 
The {\em reduced graph} $\tilde{G}$ is obtained from $G$ by removing the following edges  (recall that $\frac{1}{k}<\eps$):
\begin{enumerate}
    \item All edges with at least one end in $C_0$. There are at most $\eps n^2$ of those.
    \item All edges inside the sets $C_i,\;1\leq i\leq k$. There are at most $g^2k\leq\eps\cdot n^2$ of those.
    \item All edges between irregular pairs $(C_i,C_j)$. There are at most $\eps\cdot\binom{k}{2}\cdot g^2\leq \eps\cdot n^2$ of those.
    \item All edges that belong to $\eps$-uniform pairs with density  $s_{ij}<\eps+\eps^{\frac{1}{r-2}}$. There are at most $(\eps+\eps^{\frac{1}{r-2}})\cdot g^2\cdot\binom{k}{2}\leq 2\eps^{\frac{1}{r-2}}\cdot n^2$ of those.
\end{enumerate}
Let us upper bound the number of copies of $\grH$ that was removed in the reduction process.
There are several cases of such copies:
\begin{enumerate}
    \item Copies with at least one vertex in $C_0$. There are at most $\eps n \cdot \binom{n}{r-1} \leq \eps\cdot n^r$ of those.
    \item Copies with at least two vertices in the same set. There are at most $k\cdot \binom{g}{2}\cdot \binom{n}{r-2} \leq \eps\cdot n^r$ of those.
    \item Copies with each vertex in a different set but with (at least one) irregular couple. 
    There are at most $\eps\cdot\binom{k}{2} \cdot \binom{k}{r-2}\cdot g^r \leq \eps\cdot n^r$ of those.
    \item Copies with each vertex in a different set but with (at least one) less-than-$(\eps+\eps^{\frac{1}{r-2}})$ couple. There are at most $\binom{k}{2} \cdot \binom{k}{r-2}\cdot (\eps+\eps^{\frac{1}{r-2}}) \cdot g^r \leq 2\eps^{\frac{1}{r-2}}\cdot n^r$ of those.
\end{enumerate}
Therefore, we  have $t(G,\grH)-t(\tilde{G},\grH) \leq 5\eps^{\frac{1}{r-2}}\cdot n^r$.

Next,  note that $\tilde{G}$ is a $k$-partite graph on the sets $C_1,\cdots,C_k$ where all the pairs $(C_i,C_j)$, $1\leq i < j \leq k$, in the partition are $\eps$-uniform with densities 
\begin{align*}
    \tilde{s}_{ij} =\begin{cases}s_{ij},\;\; s_{ij} \geq \eps+\eps^{\frac{1}{r-2}},\\0,\;\;\; s_{ij} < \eps+\eps^{\frac{1}{r-2}},\text{ or $(C_i,C_j)$ is irregular pair}.\end{cases}
\end{align*}
Let us lower bound the number of copies of $\grH$ in $\tilde{G}$. We say that a vertex $v\in C_i$ is {\em typical} with respect to $C_j$, $j\neq i$, if it is connected to at least $(\tilde{s}_{ij}-\eps)\cdot|C_j|$ vertices in $C_j$. If $v\in C_i$ is not typical with respect to $C_j$ then it is {\em atypical} with respect to $C_j$. 
Note that for any $\eps$-uniform pair $(C_i,C_j)$, there are at most $\eps|C_i|$ vertices in $C_i$ that are atypical with respect to $C_j$ (if there were more, 
they would form a subset $V\subset C_i$, with size  $|V|>\eps|C_i|$ and $d(V,C_j)<\tilde{s}_{ij}-\eps$, in contradiction to $(C_i,C_j)$ being an $\eps$-uniform pair). 
Hence, given a specific $r$-let of distinct sets $C_{i_1},\cdots,C_{i_r}$, we have that there are at least $(1-(r-1)\eps)\cdot |C_{i_1}|$ vertices in $C_{i_1}$ that are typical with respect to all $C_{i_j}$, $2\leq j\leq r$.  
Let us look at a specific such vertex $v_1\in C_{i_1}$. Denote the span of $v_1$ in $C_{i_j}$ by $A_j$, and note that $|A_j|>(\tilde{s}_{i_1i_j}-\eps) |C_{i_j}|$, $\forall j\in \{2,\cdots,r\}$, and that for any $j$ such that $\tilde{s}_{ij}>0$ we get $|A_j|\geq \eps |C_{i_j}|$. 
Next, note that there are at least $(\tilde{s}_{i_1 i_2}-(r-1)\eps)\cdot |C_{i_2}|$ vertices in $A_{2}$ that are typical with respect to all $A_{j}$, $3\leq j\leq r$, thus connected to at least $(\tilde{s}_{i_2 i_j}-\eps)(\tilde{s}_{i_1 i_j}-\eps)|C_{i_j}|$ vertices in each $A_j$, $3\leq j\leq r$. 
Let us choose one such vertex, $v_2\in A_2\subseteq C_{i_2}$, and repeat the process again. 
After repeating the process $r-2$ times, we have a subset of $r-2$ vertices \{$v_1\in C_{i_1},\cdots, v_{r-2}\in C_{i_{r-2}}\}$, that are all connected to each other and to the same $\left[  \prod_{j=1}^{r-2}(\tilde{s}_{i_{j}i_{\ell}}-\eps)\right]\cdot |C_{i_\ell}|$, vertices in $C_{i_{\ell}}$, $\ell\in \{r-1,r\}$. We denote these subsets as $\tilde{A}_\ell\subset C_{i_\ell}$. Since $\tilde{s}_{ij}\geq\eps+\eps^{\frac{1}{r-2}}$ we get that
\begin{equation*}
    \prod_{k=1}^{r-2}(\tilde{s}_{i_{j}i_{\ell}}-\eps)>\eps,
\end{equation*}
hence the pair ($\tilde{A}_{r-1}$,$\tilde{A}_{r}$) has at least 
\begin{align}
(\tilde{s}_{r-1,r}-\eps)\cdot |\tilde{A}_{r-1}||\tilde{A}_{r}|& \geq \left[ \prod_{j=1}^{r-2}(\tilde{s}_{i_{j}i_{r-1}}-\eps)\right]|C_{i_{r-1}}|\\
&\hspace{-0.5cm} \times \left[ \prod_{\ell=1}^{r-2}(\tilde{s}_{i_{\ell}i_{r}}-\eps)\right]|C_{i_{r}}|
\end{align}
edges between them, each edge completes one copy of $\grH$ (with the vertices $v_1,\cdots,v_{r-2}$). 
But, there were $ \left[\prod_{j=1}^{\ell-1}(\tilde{s}_{i_{j}i_{\ell}}-(r-1)\eps\right]\cdot |C_{i_\ell}|$ options to choose the vertices $v_{\ell}$, $2\leq\ell\leq r-2$,  and $(1-(r-1)\eps)|C_{i_1}|$ options to choose $v_1$. Hence, if we let $A\triangleq\{i_1,\cdots,i_r\}$, and denote by $N^{\tilde{G}}_{A}$ the number of copies of $\grH$ over the $r$-let $C_{i_1},\cdots,C_{i_r}$ in $\tilde{G}$, we get  
\begin{align}
    N^{\tilde{G}}_{A} &\geq (1-(r-1)\eps)|C_{i_1}|
    \prod_{\ell=2}^{r-2} \left(\prod_{j=1}^{\ell-1}(\tilde{s}_{i_{j}i_{\ell}}-(r-1)\eps\right)\\
    &\hspace{1cm}\times|C_{i_\ell}|\prod_{\tilde{\ell}=r-1}^r\left( \prod_{j=1}^{\tilde{\ell}-1}(\tilde{s}_{i_{j}i_{\tilde{\ell}}}-\eps)\right)|C_{i_{\tilde{\ell}}}|\\
          &>  (1-(r-1)\eps)\prod_{i,j\in A} \left(\tilde{s}_{ij} - (r-1)\eps\right)\prod_{\ell=1}^r|C_{i_{\ell}}|\\
                    &>  \left(\prod_{i,j\in A}\tilde{s}_{ij} -r^3 \eps^{\frac{r-3}{r-2}}\right)\prod_{\ell=1}^r|C_{i_{\ell}}|,\label{eq:deviation_constant}
\end{align}
where \eqref{eq:deviation_constant} can be  derived using the relations $\log(1-\frac{1}{x})^{\frac{1}{x}}\geq e^{-\frac{1}{1-x}}$, $\forall x\in(0,1)$ and $e^{-x}\geq 1-x$, $\forall x>0$.  
Then,
\begin{align}
    t(\tilde{G},\grH)&= \frac{1}{\binom{n}{r}}\sum_{\small{\begin{subarray}{l}A\subseteq [k]\\
    |A|=r\end{subarray}}}N^{\tilde{G}}_{A}\geq \frac{n^r}{\binom{n}{r}k^r} \sum_{\small{\begin{subarray}{l}A\subseteq [k]\\
    |A|=r\end{subarray}}}\left(\prod_{i<j\in A} \tilde{s}_{ij}- r^3 \eps^{\frac{r-3}{r-2}}\right)\\
    &= 
    \left(\frac{1}{\binom{k}{r}} \sum_{\small{\begin{subarray}{l}A\subseteq [k]\\
    |A|=r\end{subarray}}}\prod_{i<j\in A} \tilde{s}_{ij}- r^3 \eps^{\frac{r-3}{r-2}}\right)c(n,k),
\end{align}
with $c(n,k)=\frac{n^r(k(k-1)\cdots (k-r+1)}{k^rn(n-1)\cdots (n-r+1)}$.
Next note that since $k<n$ we have
\begin{align}
        c(n,k)=\frac{\prod_{\ell=1}^{r-1}(1-\frac{\ell}{k})}{\prod_{\ell=1}^{r-1}(1-\frac{\ell}{n})}\leq  1,
\end{align}
and also
\begin{align}
    c(n,k)=\frac{\prod_{\ell=1}^{r-1}(1-\frac{\ell}{k})}{\prod_{\ell=1}^{r-1}(1-\frac{\ell}{n})}\geq \frac{\left(1-\frac{r-1}{k}\right)^{r-1}}{\left(1-\frac{1}{n}\right)^{r-1}}\geq \left(1-\frac{r-1}{k}\right)^{r-1}\geq 1-(r-1)^2\eps,
\end{align}
where in the last transition we used the fact the $r\ll k$, the taylor series of $(1+x)^\alpha$ and the relation $\frac{1}{k}\ll \eps$.
Then we get 
\begin{align}
    t(\tilde{G},\grH)&=\geq 
    \frac{1}{\binom{k}{r}} \sum_{\small{\begin{subarray}{l}A\subseteq [k]\\
    |A|=r\end{subarray}}}\prod_{i<j\in A} \tilde{s}_{ij}- r^3 \eps^{\frac{r-3}{r-2}}-(r-1)^2\eps
\end{align}
Note that since $\tilde{s}_{ij}$ is different from $s_{ij}$ only if $s_{ij}<\eps+\eps^{\frac{1}{r-2}}$ or if $(C_i,C_j)$ is irregular pair, which only occurs at $\eps\binom{k}{2}$ of the elements in the sum, we get that 
\begin{align*}
\Big|t(\S,\grH)-\frac{1}{\binom{k}{r}}\sum_{\small{\begin{subarray}{l}A\subseteq [k]\\
    |A|=r\end{subarray}}}\prod_{i<j\in A} \tilde{s}_{ij}\Big|\leq 3\eps^{\frac{1}{r-2}}.
    \end{align*}
Hence, after considering the copies that was lost in the reduction process, we get
\begin{align}
    t(G,\grH) &\geq t(\tilde{G},\grH) - 3\cdot \eps^{\frac{1}{r-2}}\geq t(\tilde{\S},\grH) -r^3 \eps^{\frac{r-3}{r-2}}-(r-1)^2\eps -3\cdot \eps^{\frac{1}{r-2}}\\
   &\geq t(\S,\grH)- \ctilde \eps^{\frac{1}{r-2}},\label{eq:small_eps_deviation_constant}
\end{align}
where \eqref{eq:small_eps_deviation_constant} is true for any $\eps<r^{-3}$.
The upper bound is derived in a similar way.
\end{proof}
\fi 

\subsection{Proof of Lemma~\ref{lem:S_size}}\label{appendix:proof_of_8}
To prove this this lemma we first need the following: we say that a bipartite graph $G=(V_1,V_2,E)$ if $\eps$-uniform if $(V_1,V_2)$ are an $\eps$-uniform pair. In all the following $g\triangleq \frac{n}{k}$.

\begin{lem}
All bipartite graph with $e$ edges over $[g]\times [g]$, except for a fraction of at most $2^{-g^2(2\eps^4 +\frac{4}{g})}$, are $\eps$-uniform.
 \end{lem}
 
\begin{proof}
Let $s=\frac{e}{g^2}$ and $G=(V_1,V_2)$ be a random bipartite graph with $|V_1|=|V_2|=g$, that is obtained by drawing $\text{Ber}(s)$ edges between $V_1$ and $V_2$ independently. 
Let $\mathcal{E}$ be the event that $G$ is not $\eps$-uniform, i.e., the event where there exist two subsets $A\subset V_1,B\subset V_2$, such that $|A|,|B|\geq\eps g$, and $e(A,B)<(s-\eps)|A||B|$ (denote this event by $\mathcal{E}_1$) or $e(A,B)>(s+\eps)|A||B|$ (denote this event by $\mathcal{E}_2$). 
Then, using Hoeffding's inequality and the union bound 
\begin{align}
    \Pr(\mathcal{E}) &\leq  \Pr(\mathcal{E}_1) + \Pr(\mathcal{E}_2)\\
    &\leq 2 \hspace{-0.2cm}\sum_{\ell_1=\eps g}^{g}\sum_{\ell_2=\eps g}^{g}\binom{g}{\ell_1}\binom{g}{\ell_2} 2^{-2\eps^2 \ell_1\ell_2}\\
    &\leq  2\hspace{-0.2cm}\sum_{\ell_1=\eps g}^{g}\sum_{\ell_2=\eps g}^{g}\binom{g}{\frac{1}{2}g}^2 2^{-2\eps^4  g^2}\label{line:Hoefdding3}\\ 
    &\leq  2^{-2\eps^4  g^2+2g+2\log g+1},
\end{align}
where in \eqref{line:Hoefdding3} we maximized the term $\binom{g}{\ell_i}$ using the choice $\ell_i=\frac{1}{2}g$, $i=1,2$, and the exponent using the choice $\ell_i=\eps g$.
Denote by $A_s$ the set of all bipartite graphs with exactly $s\cdot  g^2$ edges and by $A_{\mathcal{E}}$ the set of all such graphs that are not $\eps$-uniform. Recall that the $\text{Ber}(s)$ distribution is uniform over all the graphs in $A_s$ and that $\Pr(A_s)\geq 1/(g^2+1)$\cite{cover1999elements}.
Then 
\begin{align}
\frac{|A_{\mathcal{E}}|}{|A_s|}&=\frac{\Pr(A_{\mathcal{E}})}{\Pr(A_s)}\\
&\leq (g^2+1)2^{-2\eps^4  g^2+2g+2\log g+1}\\
&\leq 2^{-g^2(2\eps^4 +\frac{4}{g})}.
\end{align}
\end{proof}

\begin{proof}[Proof of Lemma \ref{lem:S_size}]
There are three degrees of freedom in constructing a graph $G$ in $\Lambda(\S,n)$: 
\begin{enumerate}
    \item Choosing all the edges between $\eps$-uniform pairs: without the $\eps$-uniformity constraint, this is equivalent to simply choosing the $s_{ij}  g^2$ edges between $C_i$ and $C_j$, which has $2^{g^2\cdot h(s_{ij})-2\log g-1}\leq |A_{s_{ij}}|\leq 2^{g^2\cdot h(s_{ij})}$ options\cite{cover1999elements}.
    Using the above claim we can deduce that the $\eps$-uniformity constraint does not change this number significantly, and we get
    $2^{g^2\cdot h(s_{ij})-2\log g-1}(1-2^{-g^2(2\eps^4 +\frac{4}{g})})\leq |A_{s_{ij},\mathcal{E}}|\leq 2^{g^2\cdot h(s_{ij})}$.
    \item Choosing the edges inside the sets: there are at most $2^{(k+1)\binom{g}{2}}\leq 2^{\eps n^2}$ options.
    \item Choosing the edges between irregular pairs: there are at most $\eps k^2$ such pairs, hence $ 2^{\eps k^2g^2}\leq 2^{\eps n^2}$ options.
\end{enumerate} 
Hence we get
\begin{align}
    2^{g^2\cdot H(S)-2k^2\log g-k^2}&(1-2^{-g^2(2\eps^4 +\frac{4}{g})})\leq|\Lambda(\S,n)| \\
    &\leq  2^{g^2\cdot H(S)+2\eps n^2 }.
\end{align}
\end{proof}
\bibliographystyle{ieeetr}
\bibliography{Szemeredi}
\end{document}